\documentclass[a4paper,english]{article}

 \usepackage{appendix}
\usepackage{stackrel}
\usepackage{tabularx}
\usepackage{amsfonts}
\usepackage{cmll}
\usepackage{amsthm}
\usepackage{amsmath}
\usepackage{proof}
\usepackage{mathdots}
\usepackage{xcolor}
\usepackage{amssymb}
\usepackage{color}
\usepackage{stackrel}

\usepackage{adjustbox}

\usepackage{tikz}
\usetikzlibrary{cd}

\usepackage{mathrsfs}
\usepackage{mathabx}

\usepackage[all]{xy}
\usepackage{lscape}
\usepackage{stmaryrd}
\usepackage{mathdots}

\usepackage{bussproofs}
\usepackage{subcaption}

\EnableBpAbbreviations

\newtheorem{example}{Example}[section]
\newtheorem{definition}{Definition}[section]

\newtheorem{remark}{Remark}[section]
\newtheorem{theorem}{Theorem}[section]

\newtheorem{lemma}[theorem]{Lemma}

\newtheorem{proposition}[theorem]{Proposition}

\newcommand{\DTT}[0]{\mathrm{dTT}}
\newcommand{\Type}[0]{\mathrm{Type}}
\newcommand{\DType}[0]{\mathrm{dType}}
\newcommand{\J}[0]{\mathsf{J}}
\newcommand{\Der}[1]{\mathsf{D}[#1]}
\newcommand{\Dist}[1]{#1^{\mathsf D}}
\newcommand{\de}[0]{\mathsf{d}}
\newcommand{\DLR}[0]{\mathsf{DLR}}
\newcommand{\Bool}[0]{\mathrm{Bool}}
\newcommand{\LET}[0]{\mathbf{let}\ }
\newcommand{\IN}[0]{\ \mathbf{in} \ }
\newcommand{\refl}[0]{\mathsf{refl}}

\newcommand{\CS}[0]{\mathsf{CS}}

\newcommand{\tri}[0]{\pitchfork}

\newcommand{\CTX}[0]{\mathsf{Ctx}}
\newcommand{\Nat}[0]{\mathrm{Nat}}

\newcommand{\Real}[0]{\mathrm{Real}}

\newcommand{\modd}[1]{\llbracket#1\rrbracket }

\newcommand{\Set}{\mathrm{Set}}
\newcommand{\Met}{\mathrm{Met}}

\newcommand{\B}[1]{\mathbf{#1}}

\newcommand{\model}[1]{\modd{#1}}
\newcommand{\nudel}[1]{\llparenthesis #1\rrparenthesis}

\newcommand{\TT}[1]{\mathtt{#1}}
\newcommand{\D}[1]{\mathscr{#1}}

\newcommand{\C}[1]{\mathcal{#1}}
\newcommand{\BB}[1]{\mathbb{#1}}
\newcommand{\BS}[1]{\boldsymbol{#1}}

\newcommand{\F}[1]{\mathfrak{#1}}

\definecolor{color0}{HTML}{4682B4}

\usepackage{amsthm}
\usepackage{hyperref}

\title{From Identity to Difference: A Quantitative Interpretation of the Identity Type} 

\author{Paolo Pistone\\ Universit\`a di Bologna, Italy\\ \url{paolo.pistone2@unibo.it}}
\date{}
\begin{document}

\maketitle

\begin{abstract}
We explore a quantitative interpretation of 2-dimensional intuitionistic type theory (ITT) in which the identity type is interpreted as a ``type of differences''. 
We show that a fragment of ITT, that we call difference type theory (dTT), yields a general logical framework to talk about quantitative properties of programs like approximate equivalence and metric preservation. To demonstrate this fact, we show that dTT can be used to capture compositional reasoning in presence of errors, since any  program can be associated with a ``derivative'' relating errors in input with errors in output. 
Moreover, after relating the semantics of dTT to the standard weak factorization systems semantics of ITT, we describe the interpretation of dTT in some quantitative models developed for approximate program transformations, incremental computing, program differentiation and differential privacy.

\end{abstract}

\section{Introduction}


\subparagraph*{From Program Equivalence to Program Differences}
In program semantics, a classical problem is to know whether two programs behave in \emph{the same} way {in all} possible contexts. Yet, in several fields of computer science, especially those involving numerical and probabilistic forms of computation (like e.g.~machine learning), it is often more important to be able to describe to which extent two programs behave in a \emph{similar}, although non equivalent, way. Hence, a crucial aspect is to be able to measure the change in the overall result that is induced by the replacement of a (say, computationally expensive) program by some (more efficient but only) approximately correct one.

These observations have motivated much research on denotational semantics involving 
\emph{metric} and \emph{differential} aspects, that is, in which one can measure differences between programs, as well as their capacity of amplifying errors. 
For instance, it has been observed that a fundamental property for a protocol to ensure \emph{differential privacy} \cite{Reed_2010, Barthe_2012, Gaboardi_2013} is that the associated program is not too \emph{sensitive} to errors; this has led to an elegant semantics \cite{Gaboardi2017} where types are interpreted as metric spaces and programs are interpreted by functions with bounded derivative (i.e.~\emph{Lipschitz-continuous} functions). More generally,  
the recent literature in theoretical computer science has seen the  
blossoming of many different
 notions of ``derivative'' for higher-order programming languages, each accounting for some differential aspect of computation: from connections with linearity (e.g.~the \emph{differential $\lambda$-calculus} \cite{difflambda, Blute2009, Manzo2010}) to incremental computation \cite{Cai2014, Picallo2019, Picallo2019b}, from higher-order automatic differentiation \cite{Pearlmutter2016, Mazza2021} to higher-order {approximate program transformations} \cite{chaudhuri, dallago, dallago2}.

%
%


Do all these differential approaches share a common logic? Is there some common notion of ``derivative'' for higher-order programs? In this paper we  
%
%
argue that a proper fragment of standard intuitionistic type theory (ITT in the following), that we call \emph{difference type theory} ($\DTT$ in short) provides a convenient framework to formalize the compositional reasoning about program derivatives found in some of these semantics. 

\subparagraph*{The Identity Type}



Our quantitative approach to intuitionistic type theory relies on a non-standard interpretation of the identity type. 
When Martin-L\"of introduced ITT \cite{Mlof75}, the identity type $I_{A}(t,u)$ was one of its main novelties. Under the Curry-Howard correspondence, the elements of $I_{A}(t,u)$ were interpreted as proofs of the fact that $t$ and $u$ denote the same object of type $A$. 
The introduction rule for the identity type constructs an object $\refl(t)\in I_{A}(t,t)$ for all $t\in A $, witnessing the equality $t=t\in A$; instead, the elimination rule provides a computational interpretation of Leibniz's indiscernibles principle by justifying a form of \emph{transport of identity}: from an equality proof $p\in I_{A}(t,u)$ and a proof $q\in \C C(t,t,\refl(t))$, one can construct a proof $\J (t,u,p,q)$ of $\C C(t,u,p)$. 

As is well-known, ITT comes in two flavors: in the \emph{extensional} version one has rules for passing to and from 
$p\in I_{A}(t,u)$ and $t=u\in A$, i.e.~saying that a proof of $I_{A}(t,u)$ exists precisely when $t=u\in A$ holds; with these rules, one can show that any element of $I_{A}(t,u)$, if any, is of the form $\refl(t)$. 
In the \emph{intensional} version these additional rules are not present (with the significant advantage that type-checking becomes decidable), and this leaves space for non-standard interpretations, as we will see. 

In more recent times, a new wave of interest in intensional ITT has spread in connection with an interpretation relating it to homotopy theory \cite{hottbook}: by exploiting the transport of identity principle, one can prove that the dependent type $I_{A}(x,y)$ carries the structure of a \emph{groupoid} (i.e.~a category with invertible arrows), and that proofs $f: A\to B$ lift to \emph{functors} $I_{A}(f): I_{A}(x,y)\to I_{B}(f(x),f(y))$ between the respective groupoids. This is the basic ground for 
 a suggestive  and well-investigated interpretation,
where $I_{A}(t,u)$ becomes the space of \emph{homotopies} between $t$ and $u$.

While the original semantics of ITT based on \emph{locally cartesian closed category} \cite{Jacobs} validates the extensionality rules (which make the homotopy interpretation trivial),  
an elegant semantics for intensional ITT has been established since \cite{Gambino2008, Awodey2009, Garner2012} based on the theory of \emph{weak factorization systems}, providing the basis for the construction of various homotopy-theoretic models.

\subparagraph*{The Type of Differences}

Without extensionality, there can be {different} ways of proving $I_{A}(t,u)$, as we saw. What if these were not seen as ways of proving that $t$ and $u$ denote the same object, but rather as ways of measuring the \emph{difference} between $t$ and $u$? The main point of this paper is to convince the reader that this idea is not only consistent, but yields some interesting new interpretation of (a fragment of) intensional ITT.

For example, in presence of a type $\Real$ of real numbers with constants $\B r$ for all $r\in \BB R$, 
we might wish to interpret an element $a\in I_{\Real}(\B r, \B s)$ as a difference between $\B r$ and $\B s$, that is, as a positive real greater or equal to $ |r-s|$. The introduction rule produces the element $\refl(\B r)=\B 0\in I_{\Real}(\B r, \B r)$, the \emph{self-distance} of $\B r$; the interpretation of the elimination rule is more delicate, as this rule makes reference to arbitrary predicates; yet, let us consider a predicate of the form $\C C(x,y,p)=I_{\Real}(f(x),g(y))$, where $x,y\in \Real$, $p\in I_{A}(x,y)$ and $f,g\in \Real\to \Real$ are two \emph{smooth} (i.e.~infinitely differentiable) functions. 
%
%
%
%
Then given a difference $a\in I_{\Real}(\B r, \B s)$ and a difference $b\in I_{\Real}(f(\B r), g(\B r))$ we can obtain a difference $c \in I_{\Real}(f(\B r), g(\B s))$ by reasoning as follows: 
first, since $g$ is smooth, by standard analytical reasoning (read: the mean value theorem) we can find some positive real $L_{r,s}$ such that $ L_{r,s} \cdot a\geq  L_{r,s}\cdot |r-s|\geq  | g(r)-g(s)|\in I_{\Real}(g(\B r), g(\B s))$; hence the operation $r,s,a \mapsto L_{r,s}\cdot a$ yields a way to \emph{transport differences} between $r$ and $s$ into differences between $g(r)$ and $g(s)$. This difference can now be used to produce a difference $\J(\B r, \B s, a, b)=b+L_{r,s} \cdot a\in I_{\Real}(f(\B r), g(\B s))$, as required by the elimination rule applied to $\C C(x,y,p)$.  

More generally, we will see that by interpreting higher-order programs as suitably ``differentiable'' maps, one can justify different kinds of ``transport of difference'' arguments, yielding quantitative interpretations of the elimination rule for the identity type.

%
%

%
%
%
%
%
%

\subparagraph*{Plan of the Paper}

In Section \ref{sec2} we introduce difference type theory $\DTT$, a fragment of ITT in which $I_{A}(t,u)$ is seen as a type of differences, and we provide a short overview of the kind of compositional differential reasoning formalizable in this system. 
In Section \ref{sec3} we introduce a notion of model for $\DTT$ (that we call a \emph{$\DTT$-category}), and we prove that any instance of this notion yields a form of weak factorization system, thus relating our semantics to the usual semantics of the identity type.

In later sections we sketch some differential models of $\DTT$. 
In Sec.~\ref{sec7} we describe the interpretation of a sub-exponential version of $\DTT$ in the metric semantics used for differential privacy \cite{Gaboardi2017}, providing a formal language to express metric preservation.
In Sec.~\ref{sec4} we show that $\DTT$ yields a natural language for differential logical relations \cite{dallago, dallago2, LICS2021}, an approach to approximate program transformations in which program differences are themselves higher-order entities. In Sec.~\ref{sec5} we provide an interpretation of $\DTT$ in models of higher-order incremental computing \cite{Cai2014}, with program differences interpreted as increments; finally, in Sec.~\ref{sec6} we show an interpretation of $\DTT$ in models of the 
differential $\lambda$-calculus \cite{difflambda, Blute2009, Manzo2010}. 

%

%
%
%

\section{Difference Type Theory}\label{sec2}

Let us start with a motivating example:
 suppose $\TT H: (\Nat \to \Real)\to \Real $ is a program that takes a function $\TT f$ from integers to reals and computes a value $\widetilde{\TT H}(\TT f(0),\dots, \TT f(N))$ depending on the first $N+1$ outputs of $\TT f$. For instance, $\TT H$ might compute some aggregated value from a time series $\TT f$ (e.g.~$\TT H$ computes the average temperature in London from a series of measures taken from $\TT f$).
%
Since measuring $\TT f$ every, say, minute might be too expensive, it might be worth considering an \emph{approximated} computation, in which $\TT f$ is only measured \emph{every $k$ minutes}, (i.e.~$\TT f$ is applied only to values $0,k,2k,\dots, \lfloor N/k\rfloor$), and each computed value is fed to $\widetilde{\TT H}$ $k$ times (this technique is well-known under the name of \emph{loop perforation} \cite{loopperf}).

What is the error we can expect for the replacement of $\TT H(\TT f)$ by its approximation?
We will show that a fragment of ITT, that we call
\emph{difference type theory} (in short $\DTT$), can be used to reason about this kind of situations in a natural and compositional way.
%
%
%

\subparagraph*{The Syntax of $\DTT$}
The fragment of ITT we consider in this paper includes two universes of types $\Type$ and $\DType$ (whose elements will be indicated, respectively, as $A,B,C,\dots$ and as $\C A, \C B, \C C,\dots$), with formation rules illustrated in Fig.~\ref{fig:types}.
In our basic language the types $A\in \Type$ are just usual \emph{simple types} (yet in our examples we will often consider extensions or variants of this language).
 The types $\C A\in \DType$ can depend on terms of some simple type; we will often refer to them as  \emph{predicates}; intuitively, an element of some predicate $\C A(t,u)$, depending on terms $t,u$ of some simple type, will be interpreted as denoting \emph{differences}, or \emph{errors}, between the terms $t,u$.

\begin{figure}[t]
\fbox{
\begin{minipage}{0.9\textwidth}
\begin{center}
\adjustbox{scale=0.8, center}{
$
\AXC{$A,B\in \Type$}
\UIC{$A \to B, A\times B \in \Type$}
\DP
\qquad \qquad
\AXC{$A\in \Type$}
\AXC{$t,u \in A$}
\BIC{$D_{A}(t,u)\in \DType$}
\DP
$ 
}

\medskip

\adjustbox{scale=0.8, center}{
$
\AXC{\phantom{$(x\in A)$}}
\noLine
\UIC{$\C A, \C B\in \DType$}
\UIC{$\C A\times \C B\in \DType$}
\DP
\qquad \qquad 
\AXC{$A\in \Type$}
\AXC{$(x\in A)$}
\noLine
\UIC{$\C C(x)\in \DType$}
\BIC{$(\Pi x\in A)\C C(x)\in \DType$}
\DP
$
}

\medskip

\adjustbox{scale=0.8, center}{
$
\AXC{$A\in \Type$}
\AXC{$(x,y\in A)$}
\noLine
\UIC{$\C C(x,y)\in \DType$}
\BIC{$(\Pi x,y\in A)(D_{A}(x,y) \to \C C(x,y))\in \DType$}
\DP
$

}
\end{center}
\end{minipage}
}
\caption{Type formation rules for $\DTT$.}
\label{fig:types}
\end{figure}

\begin{figure}[t]
\fbox{
\begin{minipage}{0.9\textwidth}

\adjustbox{scale=0.8, center}{
$
\AXC{\phantom{$(x\in A)$}}
\noLine
\UIC{$t\in A\phantom{)}$}
\UIC{$\partial(t) \in D_{A}(t,t)$}
\DP
\qquad
\AXC{$a\in D_{A}(t,u)$}
\AXC{$(x,y\in A$)}
\noLine
\UIC{$\C C(x,y)\in \DType$}
\AXC{$(x\in A)$}
\noLine
\UIC{$b(x) \in \C C(x,x)$}
\TIC{$\J (t,u,a,[x]b) \in \C C (t,u)$}
\DP
$}

\medskip

\adjustbox{scale=0.8, center}{
$
\AXC{$t\in A$}
\AXC{$(x\in A)$}
\noLine
\UIC{$b(x) \in \C C(x,x)$}
\RL{$(\beta)$}
\BIC{$\J (t,t,\partial(t),[x]b)= b[t/x] \in \C C (t,t)$}
\DP
\qquad
\AXC{\phantom{$(x\in A)$}}
\noLine
\UIC{$a\in D_{A}(t,u)$}
\RL{$(\eta)$}
\UIC{$ \J (t,u,a, [x]\partial(x)) = a \in D_{A}(t,u) $}
\DP
$}

\end{minipage}
}
\caption{Introduction, elimination and computation rules for $D_{A}(t,u)$.}
\label{fig:rules}
\end{figure}

%
%
%
%

The introduction, elimination and computation rules of $\DTT$ are those of standard intuitionistic type theory, restricted to the types of $\DTT$ (we recall them in the Appendix). We use $D_{A}(t,u)$ instead of $I_{A}(t,u)$ for the usual identity type, since we are thinking of it as a \emph{type of differences}. 
We illustrate in Fig.~\ref{fig:rules} the rules for the difference type $D_{A}(t,u)$. Terms are constructed starting from a countable set of \emph{term variables} $x,y,z,\dots$ and a countable (disjoint) set of \emph{difference variables} $\epsilon, \delta, \theta,\dots$.
The rules in Fig.~\ref{fig:rules} must be read as dependent on some \emph{context}, which, for $\DTT$, 
are of the form $(\B x\in \Phi_{0}\mid\BS\epsilon \in \Phi_{1}(\B x))$, where 
 $\B x\in \Phi_{0}= (x_{1}\in A_{1},\dots, x_{n}\in A_{n})$
is a sequence of (non-type dependent) declarations for $A_{i}\in \Type$, and  
$\BS\epsilon \in\Phi_{2}(\B x)=( \epsilon_{1}\in \C C_{1}(\B x), \dots, \epsilon_{m}\in \C C_{m}(\B x))$ is a sequence of declarations for $(\B x\in \Phi_{0})\C C_{i}(\B x)\in \DType$.

A term $(\B z\in \Phi_{1})t\in A$, where $A\in \Type$, is just an ordinary $\lambda$-term with pairing.
We call such terms \emph{program terms} and we use $t,u,v$ for them. 
We will consider variants of this basic language with other type and term primitives (e.g.~ground types like $\Nat, \Bool$ or the probabilistic monad $\Box A$), as well as a \emph{sub-exponential} variant ST$\lambda$C$^{!}$, with bounded linear types of the form $!_{k}A\multimap B$ (see \cite{Girard92, Reed_2010}), described in the Appendix.

A term $(\B x\in \Phi_{0}\mid \BS \epsilon \in \Phi_{1}) a\in \C A$, for some $(\B x\in \Phi_{0})\C A \in \DType$ belongs to the grammar
\begin{center}
$
a ::= \epsilon \mid \lambda x.a\mid at \mid \lambda xy\epsilon.a \mid atta \mid
\langle a,a\rangle\mid \pi_{1}(a)\mid \pi_{2}(a)\mid 
\partial(t) \mid \J(t,t,a,[x]a)
$
\end{center}
We call such terms \emph{difference terms} and we will use $a,b,c$ for them.

Intuitively, a term of the form $\partial(t)$ (i.e.~$\mathsf{refl}(t)$ in ITT) indicates the \emph{self-difference} of $t$.
For example, when considering semantics based on metric spaces, $\partial(t)$ will represent the null error, i.e.~$0$.
 However, in other models of $\DTT$, $\partial(t)$ needs not be zero (in fact, 0 is not even part of our basic syntax). In particular, in the models from Section 5 and 6, for a higher-order function $f\in A\to B$, $\partial(f)$ will provide a measure of the sensitivity of $f$ (in fact, in such model $\partial(f)$ coincides with the \emph{derivative} of $f$, see below).

The terms of the form $\J (t,u,a,[x]b)$ are the main computational objects (and also the least intuitive) of $\DTT$. The idea behind the quantitative interpretation of this constructor is that, given self-differences $b(x)\in \C C(x,x)$, $\J$ ``transports'' an error  $a$ between $t$ and $u$,  measured in $D_{A}$, onto an error between $t$ and $u$ 
measured in $\C C$. For example, as discussed in the introduction, $\C C(x,y)$ might be the type of differences $D_{B}(fx, gy)$, and $\J$ will thus transport a difference $a$ between $t$ and $u$ onto a difference between $ft$ and $gu$.

A fundamental application of $\J$ is the following:
for any function $f\in A\to B$, the \emph{derivative of $f$} is the following difference term
$$
\Der f:= \lambda xy\epsilon. \J (x,y,\epsilon,  [x]\partial(fx))\in
(\Pi x,y\in A)(D_{A}(x,y)\to D_{B}(fx,fy))
$$
$\Der f$ tracks errors in input into errors in output of $f$, and can thus be taken as a measure of the \emph{sensitivity} of $f$. In the models of $\DTT$ described in the following sections $\Der f$ will be interpreted by different notions of program derivative, including the ``{true}'' derivative of $f$, when the latter encodes a real-valued smooth function.

From the computation rules of $\J$ we deduce the following computation rules for derivatives:
\begin{align*}
\Der f(t ,t, \partial(t)) & = \partial(ft) \tag{$\beta\mathsf D$}\label{betad} \\
\Der{\lambda x.x}(t,u,a) & = a  \tag{$\eta\mathsf D$}\label{etad} 
\end{align*}
\eqref{betad} says that the derivative of $f$ computed on the self-distance of a point is just the self-distance of the image of the point. When $\Der f$ is seen as the ``true'' derivative, the self-distances $\partial(v)$ are just the null error $0$, and so \eqref{betad} says that the derivative computed in $0$ is $0$. 
\eqref{etad} says that the error produced in output by the identity function is just the error in input (this is in accordance with the analytical intuition too). 

Given $f\in A\to B$ and $g\in B\to C$ the composition of $\Der f$ with $\Der g$ yields a difference
of type $(\Pi x,y\in A)(D_{A}(x,y)\to  D_{C}(gf(x),gf(y))$. 
Several models of $\DTT$ will satisfy the \emph{chain rule} axiom below, which identifies the latter with the derivative of $g\circ f$:
\begin{align}
\Der{\lambda x.g(fx)} = \lambda xy\epsilon. \Der g (fx) (fy) (\Der f xy\epsilon)
\tag{$\mathsf{D}\mathrm{chain}$}\label{dchain}
\end{align}


%
%
%
%
%
%

\begin{example}\label{example:1}

For all $f\in \Nat\to \Real$, let  $ f^{*}\in \Nat\to \Real$ be defined by $ f^{*}(2i)= f(i)$ and $ f^{*}(2i+1)= f^{*}(2i)$. The loop perforation of index 2 of $\TT H( f)$ is precisely $\TT H( f^{*})$. We sketch how to construct a difference between $\TT H( f)$ and $\TT H( f^{*})$ in $\DTT$.

The step function $\Delta f (x)=|f(x+1)-f(x)|$ can be defined as $\Delta f(x)=\Der f(x+1,1)$ (where we take a difference between $x,y\in \Nat$ to be any positive real $\geq |x-y|$).  
For all $x\in \Nat$, the function $d$ with $d(2x)=0$ and $d(2x+1)=\Delta f(2x)$ yields then an element $d\in (\Pi x\in \Nat)D_{\Real}(f(x), f^{*}(x))$. 
Using this and the derivative of $\widetilde{\TT H}$ we can compute a distance
$b \in D_{\Real}( {\TT H}( f), {\TT H}( f^{*}))$ by 
$b= \Der{\widetilde{\TT H}}([ f(i),d(i)]_{i=0,\dots,N})$.

More generally, we can construct a function $c\in(\Pi f,g\in \Nat\to \Real)(D_{\Nat\to\Real}(f,g)\to D_{\Real}( {\TT H}(f), {\TT H}(g))$:   from  a distance $\epsilon\in D_{\Nat\to \Real}(f,g)$ we can deduce a function 
$e(\epsilon)\in (\Pi x\in \Nat)D_{\Real}(f(x),g(x))$ by letting $e(\epsilon)= \J (f,g,\epsilon,[x]\partial(f(x)))$, and we define $c(f,g,\epsilon)$ by replacing $d$ by $e(\epsilon)$ in $b$. 

\end{example}

\begin{example}[distance function]\label{ex:distance}

In presence of a type $\Bool$ for Booleans, with constants $\B 0, \B 1\in \Bool$ and $\mathsf{case}_{C}:\Bool \to C\to C \to C$ (with $\mathsf{Case}_{C}(\B 0, x,y)=x$ and $\mathsf{Case}_{C}(\B 1, x,y)=y$), and with a constant $\infty \in D_{\Bool}(\B 0, \B 1)$, it is possible to define a distance function 
$d_{A}\in (\Pi x,y\in A)D_{A}(x,y)$ for all simple type $A$, by letting 
$d_{A}= \lambda xy.\J ( \B 0, \B 1,\infty, [w]\partial(\mathsf{case}_{A}(w,x,y)))$.
If we admit the equational rule $\mathsf{Case}_{C}(w,x,x)=x$, then using the Equation \eqref{weak} (see below) we can deduce that $d_{A} xx $ coincides with the self-difference $\partial(x)$. 

\end{example}
\subparagraph*{Predicates in $\DTT$}
An important property of $\DTT$ is that any predicate $(\B x\in \Phi_{0})\C C(\B x)\in \DType$ is obtained by substitution from a special family of binary predicates, defined below.

\begin{definition} A predicate is said \emph{pure for $A$} if it is of the form $(x,y\in A)\C C(x,y)$
%
and one of the following holds:
\begin{itemize}
\item $\C C(x, y)=D_{A}(x,y)$;
\item $A=B_{1}\times B_{2}$ and $\C C( x,  y)= \C B_{1} (\pi_{1}(x), \pi_{1}(y))\times \C B_{2}(\pi_{2}(x), \pi_{2}(y))$, where $\C B_{1}$ is pure for $B_{1}$ and $\C B_{2}$ is pure for $B_{2}$;
\item $A=B\to C$ and $\C C( x,  y)= (\Pi z\in A)\C B(xz, yz)$, where $\C B(z,z')$ is pure for $C$;
\item $A=B\to B\to C$ and $\C C( x,  y)= (\Pi x',y'\in B)(D_{A}(x',y')\to \C B(xx'y',  yx'y'))$, where $\C B(x',y')$ is pure for $C$.

\end{itemize}
\end{definition}
%
%
%
\begin{lemma}\label{lemma:pure}
For any predicate $(\B z\in \Phi_{0})\C C(\B z)$ there exists a pure predicate $(x, y\in A)\C C^{\flat}( x,  y)$ and terms $  (\B z\in \Phi_{0})t,u\in A$ such that 
$\C C(\B z)= \C C^{\flat}( t,  u)$.
\end{lemma}

For example, the predicate $(z\in A, w\in A) (\Pi x\in B)D_{C}(f(z,x), g(z,w,x))$ is obtained from the pure predicate $(y,y'\in B\to C)(\Pi x\in B)D_{C}(yx,y'x)$ and the functions
$(z\in A, w\in A)\lambda x.f(z,x),\lambda x.g(z,w,x)\in B\to C$.
Lemma \ref{lemma:pure} will play a crucial role in defining models of $\DTT$: we will start by interpreting pure predicates,  and we will obtain the interpretation of all other predicates by a \emph{pullback} operation (corresponding to substitution).

\begin{remark}
In standard ITT, the identity type $I_{A}(\_,\_)$ yields a \emph{groupoid}\footnote{In fact, one obtains a groupoid by considering elements $p\in I_{A}(t,u)$ up to the equivalence induced by 3-dimensional homotopies in $I_{I_{A(t,u)}}(p,q)$.}, i.e.~a category with invertible arrows. In $\DTT$ one can only prove that $D_{A}(\_,\_)$ has the structure of a \emph{deductive system} (i.e.~a \emph{non-associative} category, see \cite{LambekScott}) in which for each arrow $a\in D_{A}(t,u)$ there is a ``transpose'' arrow $a^{*}\in D_{A}(u,t)$, with $\partial(t)^{*}=\partial(t)$.  
\end{remark}

\begin{remark}
In some formulation of intuitionistic type theory (e.g.~see \cite{Jacobs}) one finds a stronger version of the $\eta$-rule, which in the fragment $\DTT$ would read as follows:
\begin{align*}
\AXC{$a\in D_{A}(t,u)$}
\AXC{$(x\in A,y\in A\mid \epsilon \in D_{A}(x,y))$}
\noLine
\UIC{$c(x,y,\epsilon)\in \C C(x,y)$}
\BIC{$
\J(t,u,a,[x]c(x,x,\partial(x)))= c(t,u,a)\in \C C(t,u)$}
\DP
\tag{$\J \eta^{+}$}\label{eta+}
\end{align*}
However, in presence of \eqref{eta+} one can deduce that $a =\partial(t)\in D_{A}(t,t)$\footnote{
This is proved as follows: by letting $c(x,y,\epsilon)=\epsilon$ and $d(x,y,\epsilon)=\partial(x)$, from $c(x,x,\partial(x))=d(x,x,\partial(x))$, we deduce
$a=
c(t,t,a)=
 \J_{\C C}(t,t,a, [x]c(x,x,\partial(x)))=
 \J_{\C C}(t,t,a, [x]d(x,x,\partial(x)))=
 d(t,t,a)=
\partial(t)
$.
} holds for all $t\in A$ and $a \in D_{A}(t,t)$, hence trivializing the interpretation of $D_{A}(x,x)$. Moreover, we can see that the rule also trivializes the interpretation of $\Der f$ as the ``true'' derivative, since it implies
$\Der f= \lambda x. \partial(fx)$\footnote{It suffices to take $c(x,y,\epsilon)=\partial(fx)$.} (i.e.~$\Der f=\lambda x.0$ when $\partial(v)$ is interpreted as the null error).

However, the following instance of \eqref{eta+} is valid in all models we consider:
\begin{align*}
\J (t,u,a, [x]b) = b \qquad (x\notin \mathrm{FV}(b)) \tag{$\J w$}\label{weak}
\end{align*}


\end{remark}


\subparagraph*{Function Extensionality}

%

In ITT the {function extensionality} axiom essentially asserts that from a proof that $f$ and $g$ send identical points into identical points, one can construct a proof that $f$ is identical to $g$. 
In the following sections we will consider models which satisfy two variants of this axiom, namely 
\begin{align*}
 D_{A\to B}(f,g) & \equiv  (\Pi x\in A)D_{B}(f(x),g(x)) \tag{$
\mathsf{FExt1}$}\label{fext1}\\
D_{A\to B}(f,g) & \equiv  (\Pi x,y\in A)(D_{A}(x,y)\to D_{B}(f(x),g(y)))\tag{$
\mathsf{FExt2}$}\label{fext2}
\end{align*}
Axioms \eqref{fext1} (resp.~\eqref{fext2}) says that a difference between two functions $f,g\in A\to B$ is the same as a map from a point $x\in A$ into a difference between $f(x)$ and $g(x)$ in  $B$ (resp.~a map from a difference $\epsilon $ between two points of $A$ into differences between their respective images).
Observe that, without these axioms, one can still construct programs
\begin{align*}
\mathsf E_{1}& \in (\Pi f,g\in A\to B)(D_{A\to B}(f,g)\to (\Pi x\in A)D_{B}(f(x),g(x)))\\
\mathsf{E}_{2} &\in(\Pi f,g\in A\to B)(D_{A\to B}(f,g) \to (\Pi x,y\in A)(D_{A}(x,y)\to D_{B}(f(x),g(y))))
\end{align*}
given by $\mathsf E_{1}= \lambda fg\phi x. \J (f,g,\phi,  [h]\lambda x.\partial(f(x)))$ and $\mathsf E_{2}= \lambda fg\delta xy\epsilon. \J( f,g,\varphi, [h]\Der{h} )$. 

Moreover, some of the models we consider will also satisfy the axiom below 
\begin{equation}
D_{A\times B}(t,u) \equiv D_{A}(\pi_{1}(t),\pi_{1}(u))\times D_{B}(\pi_{2}(t),\pi_{2}(u)) 
\tag{$\mathsf{CExt}$} \label{cext}
\end{equation}
stating that a difference between pairs is a pair of differences. Even without \eqref{cext}, one can  construct terms $\mathsf C_{1},\mathsf C_{2}$ to and from the types above (yet they do not define an isomorphism).

In presence of one or more of the extensionality axioms, it makes sense to consider further computational rules for the operators $\J$ and $\mathsf D$ (for instance, the rule stating that for a higher-order function $f\in A\to B$,  $\partial(f)=\Der f$), that we discuss in the Appendix.

\section{Models of $\DTT$ and Weak Factorization Systems}\label{sec3}


The by now standard semantics of the identity type is based on \emph{weak factorization systems} (in short, WFS). A WFS is a category endowed with two classes of arrows $\C L$ and $\C R$, such that any arrow factorizes as the composition of a $\C L$-arrow and a $\C R$-arrow (the typical example is $\Set$, with $\C L$ being the class of surjective functions and $\C R$ the class of injective functions).

The goal of this section is to introduce a workable notion of model for $\DTT$, as formal basis for the concrete models illustrated in the next sections, and to relate it to the standard WFS semantics of ITT.
We will first present a basic setting, that we call a \emph{$\DTT$-category}, which allows for the interpretation of $\DTT$ (and roughly follows \cite{Garner2012}). We then introduce a slight variant of WFS, that we call \emph{$U$-WFS}, where $U$ is some monoidal functor. This variant is adapted to the ontology of $\DTT$, where one has two distinct families of terms, and only requires that the ($U$-image of the) arrows from the first family factor through the arrows of the second family.
We finally show that any $\DTT$-category gives rise to a $U$-WFS.

The fundamental example of a $\DTT$-category will be the \emph{context category} of $\DTT$, that is, the category $\CTX$ with objects being contexts and arrows $(\B x\in \Phi_{0}\mid \BS\epsilon \in \Phi_{1}(\B x))\to (\B y\in \Psi_{0}\mid \BS\delta\in \Psi_{1}(\B y))$ being sequences $(\B t\mid \B a)$ of ($\beta\eta$-equivalence classes of) terms such that $(\B x\in \Phi_{0})t_{i}\in A_{i}$ and $(\B x\in \Phi_{0}\mid\BS\epsilon\in \Phi_{1}(\B x))a_{j}\in \C C_{j}(\B t)$ holds for all $A_{i}\in\Type$ occurring in $\Psi_{0}$ and $\C C_{j}(\B x)\in \DType$ occurring in $\Psi_{1}$.
We let $\CTX_{0}$ be the full subcategory of $\CTX$ made of contexts of the form $(\B x\in \Phi_{0}\mid )$, and 
$\iota:\CTX_{0}\to \CTX$ indicate the inclusion functor.

When considering the sub-exponential simply typed $\lambda$-calculus ST$\lambda$C$^{!}$ as base language, we let $\CTX_{0}^{!} $ indicate the category of ST$\lambda$C$^{!}$-typed terms and 
$H: \CTX_{0}^{!}\hookrightarrow\CTX_{0}\stackrel{\iota}{\to}\CTX$ indicate the associated embedding inside $\CTX$ (where $\CTX_{0}^{!} \hookrightarrow \CTX_{0}$ corresponds to the ``forgetful'' embedding of ST$\lambda$C$^{!}$ inside ST$\lambda$C - for more details, see the Appendix). 

Observe that the category $\CTX_{0}$ is cartesian closed, while $\CTX_{0}^{!}$ is  symmetric monoidal closed and $\CTX$ is only cartesian. Hence, the basic data to interpret $\DTT$ will be given by a strict monoidal functor $U:\BB C_{0}\to \BB C$ between a symmetric monoidal closed category $\BB C_{0}$ (interpreting either ST$\lambda$C or ST$\lambda$C$^{!}$) and a cartesian category $\BB C$ (interpreting the difference terms). We will use $\Gamma.\Delta$ for the monoidal product of $\BB C_{0}$.

While $\BB C_{0}$ only accounts for simple types, $\BB C$ needs to have enough structure to account for type dependency: for all object $\Gamma$ of $\BB C_{0}$, we consider 
a collection $\D P(\Gamma)$ of \emph{predicates over $\Gamma$} such that,
for all $P\in \D P(\Gamma)$, there exists an object $\Gamma\mid P$ of $\BB C$ and an arrow $\pi_{\Gamma}:\Gamma\mid P \to U\Gamma$ called the \emph{projection} of $P$. We also require that for any predicate $P\in \D P(\Gamma)$ and $f:\Delta\to \Gamma$, the pullback $(Uf)^{\sharp}(\Gamma\mid P)$ exists and is generated by some object $ f^{\sharp}P\in \D P(\Delta)$: 

\adjustbox{center, scale=0.9}{$
\begin{tikzcd}
{\Delta\mid f^{\sharp}P} \ar{d}[left]{\pi_{\Delta}} \ar{rr}{f^{+}} & & {\Gamma\mid P} \ar{d}{\pi_{\Gamma}} \\
  {U\Delta} \ar{rr}[below]{Uf} & &{U\Gamma}
\end{tikzcd}
$}
\medskip

\noindent 
Moreover, we require that the equalities 
$
\mathrm{id}^{\sharp}P=P$, $(g\circ f)^{\sharp}p=g^{\sharp}(f^{\sharp}P)$, $
\mathrm{id}^{+}=\mathrm{id}$, $ (g\circ f)^{+}= g^{+}\circ f^{+}$ all hold.
 In the case of $\CTX$, $\D P(\Phi_{0})$ is the set of predicates $(\B x\in \Phi_{0})\C C(\B x)\in \DType$. Given a predicate $(x,y\in A)\C C(x,y)$ and simply typed terms $(t,u):(\B y\in \Psi_{0}) \to (x,y\in A)$, the pullback $(t,u)^{\sharp}\C C$ corresponds to the predicate $(\B y\in \Psi_{0})\C C(t(\B y), u(\B y))$.

Given predicates $P\in \D P(\Gamma)$ and $Q\in \D P(\Delta)$, we indicate an arrow 
$h\in \BB C(\Gamma \mid P, \Delta \mid Q)$
as $(h_{0}\mid h_{1})$ if $h=h_{0}^{+}\circ h_{1}$, for some $h_{0}\in \BB C_{0}( U\Gamma, U\Delta) $ and  $h_{1}\in \BB C(\Gamma\mid P, \Gamma\mid h_{0}^{\sharp}Q)$ occurring in a commuting diagram as below.

\adjustbox{center, scale=0.9}{$
\begin{tikzcd}
\Gamma \mid P \ar{r}{h_{1}} \ar{d}[left]{\pi_{\Gamma}} & \Gamma \mid h_{0}^{\sharp}Q \ar{d}{\pi_{\Gamma}} \ar{r}{h_{0}^{+}} & \Delta\mid Q \ar{d}{\pi_{\Delta}}\\
U\Gamma \ar[-,double]{r} & U\Gamma  \ar{r}{Uh_{0}}& U\Delta
\end{tikzcd}
$}
\medskip

\noindent In $\CTX$ this precisely says that an arrow $(\B t\mid \B a): (\B x\in \Phi_{0}\mid \BS\epsilon\in \Phi_{1}(\B x))\to (\B y\in \Psi_{0}\mid \BS\delta\in \Psi_{1}(\B y))$ is composed of arrows $\B t\in \Phi_{0}\to  \Psi_{1}$ and 
$\B a\in  \Phi_{1}(\B x)\to \Psi_{1}(\B t(\B x))$. 

To handle the difference types we need to make some further requirements. First, we consider a sub-family $\D P^{\flat}(\Gamma)\subseteq \D P(\Gamma.\Gamma)$ of binary predicates, that we call \emph{pure} predicates, which \emph{generates} the family $\D P(\_)$, in the sense that for all object $\Gamma$ and predicate 
$P\in \D P(\Gamma)$ there exists an object $\Delta$, a pure predicate $P^{\flat}\in \D P^{\flat}(\Delta)\subseteq \D P(\Delta.\Delta)$ and $f\in \BB C_{0}(\Gamma, \Delta.\Delta)$ such that $P= f^{\sharp}(P^{\flat})$. 
In the case of $\CTX$ this is precisely what is asserted by Lemma \ref{lemma:pure}.

For any $\Gamma$, we require a choice of a pure predicate 
 $\Dist{\Gamma}\in \D P^{\flat}(\Gamma)$. The introduction rule requires the existence of an arrow $r_{\Gamma}: U\Gamma \to( \Gamma.\Gamma\mid \Dist\Gamma)$ such that $\pi_{\Gamma.\Gamma}\circ r_{\Gamma} $ coincides with the diagonal $\delta_{U\Gamma}: U\Gamma \to U(\Gamma.\Gamma)=U\Gamma\times U\Gamma$. In $\CTX$ $\Dist{(\Phi_{0})}$ is $(\B x,\B y\in \Phi_{0}\mid \BS\epsilon \in D_{\Phi_{0}}(\B x, \B y))$, where $D_{\Phi_{0}}(\B x, \B y)$ is the list of all $D_{A_{i}}(x_{i},y_{i})$, for $A_{i}$ occurring in $\Phi_{0}$, and $r_{\Phi_{0}}$ is given by $(\B x, \B x\mid \partial(\B x))$ (where $\partial(\B x)=\langle \partial(x_{1}),\dots, \partial(x_{k})\rangle$).
Actually, in order to handle contexts properly, we must consider a slightly more complex condition (see the Appendix).
To handle the elimination rule, 
 for any binary predicate $P=  f^{\sharp}P^{\flat}\in \D P(\Gamma.\Gamma)$ 
%
and commutative diagram

\adjustbox{center, scale=0.9}{$
\begin{tikzcd}
U\Gamma   \ar{d}[left]{r_{\Gamma}} \ar{rr}{c} & & \Gamma. \Gamma \mid P \ar{d}{ \pi_{\Gamma.\Gamma}} \\
\Gamma.\Gamma\mid \Dist{\Gamma} \ar{rr}[below]{  \pi_{\Gamma.\Gamma}} & & U\Gamma\times U\Gamma
\end{tikzcd}
$}
\medskip

\noindent
we require the existence of a diagonal filler $j: (\Gamma.\Gamma\mid \Dist\Gamma) \to (\Gamma.\Gamma\mid P)$ making both triangles commute. 
In $\CTX$, $P=f^{\sharp}P^{\flat}$ is a predicate $(\B x, \B y\in \Phi_{0})\C C(\B x, \B y)= \C C^{\flat}(f_{1}(\B x, \B y), f_{2}(\B x,\B y))$, $c$ is of the form $( \B x, \B x \mid c'(\B x))$, where $c'(\B x)\in \C C(\B x, \B x)$, and a diagonal filler is provided by $j=(\B x, \B y\mid  \J(\B x, \B y, \BS\epsilon, [\B x]c'))$.
%
%
The commutation of the upper triangle $j\circ r_{\Gamma}= c$ coincides then with the $\beta$-rule.
The validity of the $\eta$-rule corresponds to the fact that, when $P=\Dist \Gamma$, $f=\pi_{1}, g=\pi_{2}$ and $c=r_{\Gamma}$, $j$ coincides with the identity arrow $ \mathrm{id}_{\Gamma.\Gamma\mid \Dist \Gamma}$. We will not require the $\eta$-condition in general.
Again, to handle contexts and substitutions properly, we must consider a slightly more complex construction, together with a few coherence conditions for $r_{\Gamma}$ and $j$ (see \cite{Awodey2009, Garner2012}), but we discuss these more technical aspects in the Appendix.

Finally, we must require that $\BB C$ has enough structure to interpret the dependent products present in the fragment $\DTT$; we describe this structure in the Appendix.

%

%

We let a \emph{$\DTT$-category} be a strict monoidal functor $U:\BB C_{0}\to \BB C$ together with collections of predicates $\D P(\_), \D P^{\flat}(\_)$ and of difference structures $(\Dist{\_}, r_{\_}, j_{\_, \_,\_})$ satisfying the properties above. 
The following proposition assures that one can interpret $\DTT$ in any $\DTT$-category.
\begin{proposition}\label{prop:interpretation}
For any $\DTT$-category $U:\BB C_{0}\to \BB C$, if $\BB C_{0}$ is cartesian closed, any map $m$ from base type variables to $ \mathrm{Ob}(\BB C_{0})$ extends into 
functors $\model{\_}_{m}: \CTX_{0}\to \BB C_{0}$ and $\nudel{\_}_{m}:\CTX\to \BB C$, satisfying $U\circ \model{\_}_{m}= \nudel{\_}_{m}\circ \iota$, and 
%
%
%
 preserving all relevant structure. 
 If $\BB C_{0}$ is symmetric monoidal closed, the same holds with $\CTX_{0}$ replaced by $\CTX_{0}^{!}$ and $\iota$ replaced by $H$.
\end{proposition}

To conclude our general presentation of the semantics of $\DTT$, we show how it relates to WFS. 
We recall that, given a category $\BB C$ and two arrows $f\in \BB C(A, B)$ and $g\in \BB C(C, D)$, $f$ is said to have the \emph{left-lifting property} with respect to $g$ (and $g$ is said to have the \emph{right lifting property} with respect to $f$), if for every commutative diagram

\adjustbox{center, scale=0.9}{$
\begin{tikzcd}
A \ar{d}[left]{f} \ar{rr}[above]{h} &  &C \ar{d}[right]{g} \\
B \ar{rr}[below]{k} & & D
\end{tikzcd}
$}
\medskip

\noindent 
there exists a diagonal filler $j\in \BB C(B,C)$ making both triangles commute.
Given a set $\D S$ of arrows in a category, we let $\D S^{\tri}$ (resp $^{\tri}\D S$) indicate the set of arrows $g$ such that any arrow in $\D S$ has the left (resp.~right) lifting property with respect to $g$. 
The following notion generalizes usual WFS:
\begin{definition}
Let $U:\BB C\to \BB D$ be a functor. An \emph{$U$-weak factorization system} (in short, $U$-WFS) \emph{for $\BB C$ inside $\BB D$} is a pair of classes of maps $(\C L, \C R)$ of $\BB D$ such that (1) for every morphism $f$ of $\BB C$, $Uf=p_{f}\circ i_{f}$, with $i_{f}\in \C L$ and $p_{f}\in \C R$, and (2) $\C L^{\tri}=\C R$ and $\C L= ^{\tri}\C R$.

\end{definition} 

Observe that a weak factorization system in the usual sense is just a $\mathrm{Id}$-WFS.

When a functor $U:\BB C_{0}\to \BB C$ yields a $\DTT$-category, it is possible to construct a $U$-WFS $(\C L_{\D P}, \C R_{\D P})$ for $\BB C_{0}$ inside $\BB C$ by letting $\C L_{\D P}= ^{\tri}\D P^{*}$ and $\C R_{\D P}=\C L^{\tri}$, where $\D P^{*}$ is made of all arrows obtained by composing the arrows $\pi_{\Gamma}: \Gamma\mid P\to U\Gamma\times U\Gamma$, for all $P\in \D P^{\flat}(\Gamma)$ with projections in $\BB C$. 
The $U$-factorization of an arrow $f: \Gamma \to \Delta$ in $\BB C_{0}$ is given by $p_{f}\circ i_{f}$, where $p_{f}=\pi_{2}\circ \pi_{\Gamma.\Delta}: \Gamma.\Delta\mid \Delta_{f}\to \Delta$, $\Delta_{f}$ is a suitable pullback, and $i_{f}$ is the arrow obtained by the universality of pullback in the diagram below:
\begin{center}
\adjustbox{scale=0.9}{
$
\begin{tikzcd}
U\Gamma\ar[bend right]{ddr}[below]{\langle U\Gamma,Uf\rangle}\ar[dashed]{rd}{r_{f}} \ar{rr}{Uf} & &U \Delta \ar{d}{r_{\Delta}} \\
 & \Gamma.\Delta\mid \Delta_{f} \ar{d}{\pi_{\Gamma.\Delta}}
 \ar{r}{} & \Delta.\Delta\mid \Dist\Delta \ar{d}{\pi_{\Delta.\Delta}}\\
&U \Gamma\times U\Delta \ar{r}{Uf.U\Delta} & U\Delta\times  U\Delta
\end{tikzcd}
$}
\end{center}
where $\Delta_{f}=(f.\Delta)^{\sharp}(\Dist \Delta)$.
To show that $r_{f}\in \C L_{\D P}$ we must rely on the difference structure in an essential way: the required diagonal filler is obtained by an arrow of the form $j$. We describe this construction (which follows the argument from \cite{Gambino2008}) in the Appendix.

\begin{theorem}\label{thm:wfs}
For any $\DTT$-category $U:\BB C_{0}\to \BB C$, with collections of predicates $\D P(\_), \D P^{\flat}(\_)$, the pair $(\C L_{\D P}, \C R_{\D P})$ forms a $U$-WFS of $\BB C_{0}$ inside $\BB C$.
\end{theorem}

\section{Metric Preservation}\label{sec7}
%
%
%

We start our parade of models of $\DTT$ by considering metric models focusing on program \emph{sensitivity}. 
In several situations it is important to know that a program is not too sensitive to small changes in the input. A key example is differential privacy: if $f: \mathrm{db}\to \Real$ is a program 
producing some aggregated information from some database $\mathrm{db}$ (e.g.~$f$ outputs the percentage of LGBTIQ+ people 
among the students of a given university), we wish the result of $f$ not to depend too much on any single item of $\mathrm{db}$, so that information about single individuals cannot be leaked from the outputs of $f$.

A standard way to capture sensitivity is through the \emph{Lipschitz}-condition: a function $f$ between metric spaces $(X,a)$ and $(Y,b)$ is $r$-Lipschitz, for some positive real $r$, when it satisfies $b(f(x),f(y))\leq r\cdot a(x,y)$ for all $x,y\in X$. It is well-known that from a $r$-Lipschitz function $f:\mathrm{db}\to \Real$ one can obtain, by adding \emph{Laplace-distributed noise}, a randomized function $f^{*}:\mathrm{db}\to \Real$ which is $r\epsilon$-differentially private.\footnote{Formally, this means that for all two inputs $x,x'\in \mathrm{db}$ that differ by at most one parameter and for all $S\subseteq \BB R$, the probability $P[f^{*}(x)\in S]$ that $f^{*}(x)\in S$ is bounded by $e^{r\epsilon}\cdot P[f^{*}(x')\in S]$.}

The type system $\mathrm{Fuzz}$ \cite{Reed_2010} was designed to ensure that well-typed programs correspond to Lipschitz functions. Formally, it is a variant of \emph{bounded linear logic} \cite{Girard92}, i.e.~an affine simply typed $\lambda$-calculus with a bounded exponential $!_{r}A$, where a program $f\in \ !_{r}A\multimap B$ corresponds to a $r$-Lipschitz function. We consider here a basic fragment ST$\lambda$C$^{!}$ of $\mathrm{Fuzz}$ (described in the Appendix).

$\mathrm{Fuzz}$ admits a natural and simple semantics in the symmetric monoidal closed category $\Met$ of metric spaces and \emph{non-expansive} (i.e.~1-Lipschitz) maps (with monoidal product $(X.a)\otimes(Y.b)= (X\times Y, a+b)$). In particular, the bounded exponential $!_{r}A$ is interpreted as the \emph{re-scaling} $rX$ of a metric space $(X,a)$ (i.e.~with $ra(x,y)=r\cdot a(x,y)$), so that a non-expansive map from $rX$ to $Y$ is the same as a $r$-Lipschitz function from $X$ to $Y$.

We construct a model of a variant of $\DTT$ where we take program terms to be ST$\lambda$C$^{!}$-typable terms. The resulting semantics will associate each simple type $A$ with some metric space, and will interpret $D_{A}(t,u)$ as the set of positive reals greater than the distance between $t$ and $u$. 
Any difference term $a$ is interpreted by a function yielding positive real numbers in output; in particular the self-differences $\partial(t)$ correspond to $0$, and for any non-expansive function $f$,  $\Der{f}$ will correspond to the  
map $x,y,\epsilon\mapsto   \epsilon$, ensuring metric preservation.

We provide a sketch of the $\DTT$-category structure of the forgetful functor $U:\Met \to \Set$, described in more details in the Appendix. For any metric space $(X,a)$, the pure predicates $\D P^{\flat}(X)$ are all (pseudo-)metric spaces $P=(X,b)$ over $X$, with $(X\times X\mid P)=\coprod_{x,y\in X}\widetilde b(x,y)$, where $\widetilde b(x,y)=\{r\mid b(x,y)\leq r\}$, with projection $\pi_{X}:(X\otimes X\mid P)\to X\times X$ given by $\pi_{X}(\langle \langle x,y\rangle, s\rangle)=\langle x,y\rangle$.\footnote{Actually, to handle the higher-order structure, we must consider \emph{parameterized} (pseudo-)metric spaces over $X$, see the Appendix.}
%
%
Moreover, for any metric space $(Y,b)$, $\D P(Y)$ is made of pullbacks  $\langle f,g\rangle^{\sharp}P$, where $P=(X,a)\in \D P^{\flat}(X)$ and $\langle f,g\rangle\in \Met (Y, X\otimes X)$, where $(Y\mid \langle f,g\rangle^{\sharp}P)=\coprod_{y\in Y}a(f(y),g(y))$, with projection $\pi_{Y}:(Y\mid \langle f,g\rangle^{\sharp}P)\to Y$ given by $\pi_{Y}(\langle y,s\rangle)=y$.

For any metric space $(X,a)$, the pure predicate $\Dist X\in \D P^{\flat}(X)$ is $(X,a)$ itself, with 
$r_{X}: X\to (X\otimes X\mid \Dist X)$ given by  
  $r_{X}(x)=\langle \langle x,x\rangle, 0\rangle$. For any binary predicate $P=\langle f,g\rangle^{\sharp}P^{\flat}\in \D P(X\otimes X)$ (with  
$ (X\otimes X\mid P)=  \coprod_{x,y\in X} \widetilde b(f(x,y),g(x,y))$, for some metric space $(Y,b)$ and $\langle f,g\rangle \in \Met(X\otimes X, Y\otimes Y)$),  %
and for any function $c: X\to
(X\otimes X\mid P)$ 
with $c(x)=\langle \langle x,x\rangle, c'(x)\rangle$, we can define a diagonal filler $j: (X\otimes X\mid \Dist X) \to (X\otimes X\mid P)$ by  
\begin{center}
$
j(\langle \langle x,y\rangle,r\rangle)=\langle \langle x,y\rangle,
  c'(x)+r\rangle$
\end{center}
In fact, from
from $\langle f,g\rangle \in \Met(X\otimes X, Y\otimes Y)$ we deduce $b(f(x,y),f(x',y'))+b(g(x,y),g(x',y'))\leq a(x,x')+b(y,y')$. Hence, from 
 $c'(x)\in \widetilde b(f(x,x), g(x,x))$ and $b(f(x,y),f(x,x))+b(g(x,x),g(x,y))\leq a(x,x)+a(x,y)\leq r$, 
%
 we deduce $(j(\langle \langle x,y\rangle,r\rangle))_{2}\in \widetilde b(f(x,y),g(x,y))$ by the triangular law.  
The validity of the $\beta$-rule follows from $((j\circ r_{X})(x))_{2}=c'(x)$; moreover, when $f(x,y)=x$, $g(x,y)=y$ and $c'(x)=0$, $j(\langle \langle x,y\rangle,\epsilon\rangle)=\langle \langle x,y\rangle,\epsilon\rangle$, so the
  $\eta$-rule is also valid.

\begin{remark}
When $h\in \Met(rX,Y)$ interprets some $r$-Lipschitz program $t\in \ !_{r}A\multimap B$, the predicate $D_{B}(tx,ty)$ corresponds to the pullback $\langle h\circ \pi_{1},h\circ \pi_{2}\rangle^{\sharp}(X,a)
\in \D P(rX\otimes rX)$, and the derivative $\Der t$ is interpreted then by the map $x,y,\epsilon \mapsto \epsilon$, as desired.
\end{remark}
%
\begin{remark}
From Theorem \ref{thm:wfs} it follows  that given metric spaces $(X,a)$ and $(Y,b)$, any $f\in \Met(r X, Y)$ factors as $X \stackrel{i_{f}}{\to} \coprod_{x\in X, y\in Y}\widetilde b(f(x),y)\stackrel{p_{f}}{\to} Y$, where 
 $i_{f}(x)=\langle \langle x, f(x)\rangle, 0\rangle$.
\end{remark}
\begin{example}
 $\DTT$ can be used to formalize meta-theoretical reasoning about $\mathrm{Fuzz}$ as discussed in \cite{Reed_2010, Gaboardi2017}. For instance, we might extend simple types with the \emph{probability monad} $\Box A$, adding  suitable primitives. Then, by interpreting the type $D_{\Box A}(t,u)$ with the metric $ d(\delta_{1},\delta_{2})=\frac{1}{\epsilon}\cdot \left(\sup_{x\in A}\left | \ln \left (\frac{\delta_{1}(x)}{\delta_{2}(x)}\right) \right |\right)$ (with $\delta_{1},\delta_{2}$ distributions over $A$), for any randomized program $f\in \  !_{r}A\to \Box B$, the statement
$\Der f \in (\Pi x,y\in A)(D_{A}(x,y)\to D_{\Box B}(f(x),f(y)))$ expresses that
for all $x,y\in A$, $\epsilon \geq d(x,y)$ and $b\in B$, $\left |P[f(x)=b]- P[f(y)=b]\right | \leq e^{r\epsilon}$, that is, that
 $f$ is a $r\epsilon$-differentially private function.
\end{example}

\begin{example}\label{example:2}
Suppose $\TT H(f)$ computes the average of the simulations $ f(0),\dots,  f(N)$, i.e.~${\TT H}(f)= \frac{1}{N+1}\cdot \sum_{i=0}^{N}f(i)$; note that $\widetilde{\TT H}(\vec x)=\frac{1}{N+1}\cdot \sum_{i=0}^{N}x_{i}: \Real^{N+1} \to \Real$ is $\frac{1}{N+1}$-Lipschitz, and thus $\Der{\widetilde{\TT H}}([x,y]_{i=0,\dots, N})=\frac{1}{N+1}\cdot \sum_{i}y_{i}$. We deduce then that for all $r$-Lipschitz function $f\in \ !_{r}\Nat\multimap \Real$, the perforation error $b\in D_{\Real}(\TT H( f), \TT H(f^{*}))$ computed in Example \ref{example:1} corresponds to 
$ \frac{1}{N+1}\cdot \sum_{i=0}^{\lfloor N/2\rfloor}r\cdot\Der{f}(2i,1)= \frac{ \lfloor N/2\rfloor+1}{N+1} \cdot r$.

\end{example}




%
%
%
%
%
%
%
%
%
%
%
%
%
%
%
%
%
%

\section{Differential Logical Relations}\label{sec4}

When studying approximate program transformations like loop perforation, the Lipschitz condition is often too restrictive. In fact, even basic operations of the simply typed $\lambda$-calculus can make this property fail: while the binary function $f(k,x)=k\cdot x: \Real^{2}\to \Real$ is $|k|$-Lipschitz in $x$ for all $k\in \BB R$, the unary function $g(x)=x^{2}$ obtained by ``contracting'' the variables $k$ and $x$ already fails to be Lipschitz. In fact, the distance between $g(x)$ and $g(x+\epsilon)$ is bounded by $2|x|\epsilon+\epsilon^{2}$, hence not proportional to $\epsilon$.
Indeed, this kind of issues is due to the fact that $\Met$ is \emph{not} a cartesian closed category, that is, a model of full ST$\lambda$C, but only of its sub-exponential variant $\mathrm{Fuzz}$.

%
%
%

The theory of \emph{differential logical relations} \cite{dallago,dallago2, LICS2021} (in short, DLR), has been developed to overcome this kind of problems when investigating approximate transformations in ST$\lambda$C.
A DLR is a ternary relation $\rho\subseteq X\times L\times X$ relating the elements of some set $X$ with the values of some complete lattice $L$ of ``errors over $X$''; intuitively, $\rho(x,\epsilon,y)$ is to be read as the fact that the error of replacing $x$ by $y$ is bounded by $\epsilon$. As the name suggests, DLR generalize usual logical relations, which can be seen as DLR where $L$ is the Boolean lattice $\{0<1\}$.
Yet, due to the arbitrary choice of $L$, a distance between two programs needs not be a Boolean nor a positive real (as in metric semantics); typically, a distance between two functional programs is itself a function, tracking distances in input into distances in output.

Since $L$ is a complete lattice, for all $x,y\in X$, one can define a distance function
$\| \_,\_\|_{\rho}: X\times X\to L$ where $\| x,y\|= \inf\widetilde\rho(x,y)$, with $\widetilde\rho(x,y)=\{\epsilon \in L\mid \exists \delta \leq \epsilon \text{ s.t. }\rho(x,\delta,y)\}$.  For instance, if we consider the DLR
$(\BB R, \BB R_{\geq 0}^{\infty},\rho_{\mathrm{Euc}})$, where $\rho_{\mathrm{Euc}}(r,v,s)$ holds iff $v\geq |r-s|$, the associated distance function is the Euclidean metric. 
However, $\|x,y\|_{\rho}$ needs not be a metric in general: first of all, the self-distances $\| x,x\|_{\rho}$ (that we note simply as $ \| x\|_{\rho}$) need not be zero (i.e.~the bottom element of $L$); moreover, $\| \_,\_\|_{\rho}$ needs not satisfy the usual \emph{triangular law} of metric spaces (for a detailed comparison between DLR and - generalized \cite{Stubbe2014} - metric spaces, see \cite{LICS2021}).
Here we will restrict our attention to \emph{separated} DLR, i.e.~such that $\| i\|_{\rho}\in\widetilde\rho(i,j)$ (or $\|j\|_{\rho}\in \widetilde\rho(i,j)$) implies $i=j$. 

%



A map of DLR $(X,L,\rho)$ and $(Y,M,\mu)$ is given by a function $f: X\to Y$ (hence no Lipschitz or other continuity conditions are asked) together with an auxiliary map $\varphi: X\times X\times  L\to Y$ which, intuitively, tracks errors in input into errors in output; more formally, $\varphi$ must satisfy, for all $ x,y\in X$ and $ \epsilon\in L$, that if $ \rho(x,\epsilon,y)$ holds, then both $\mu(f(x), \varphi(x,y,\epsilon), f(y))$ and $\mu(f(x), \varphi(x,y,\epsilon), f(y))
$ also hold.
(Separated) DLR and their maps form a cartesian closed category $\DLR$ (see the Appendix, and \cite{dallago, LICS2021} for further details).

The presence of the auxiliary map $\varphi$ is what ensures the ``transport'' of errors: if $\rho(t,\epsilon,u)$ holds for some $t,u$ of type $A$ and the context $\TT C[\ ]:A\to B$ admits an auxiliary map $\varphi$, then  $\varphi(t,u,\epsilon)$ provides an error bound between $\TT C[t]$ and $\TT C[u]$.
For example, to the function $g(x)=x^{2}$ one can associate the auxiliary map $\varphi_{g}(x,y,\epsilon)= 2|x|\epsilon+\epsilon^{2}$. 

%

%
%


To model $\DTT$ in terms of DLR we will interpret $D_{A}(t,u)$ as the set $\widetilde\rho(t,u)$ of differences between $t$ and $u$; hence, the self-differences $\partial(t)$ will correspond to $\|t\|_{\rho}$, and $\Der{f}$ will provide each program $f$ with the auxiliary map
$x,y,\epsilon \mapsto \sup\{ \| f(x),f(z)\|_{\mu}\mid z\in X\land \rho(x,\epsilon,z)\}$.
Moreover, due to the higher-order structure of DLR (recalled in the Appendix) the DLR models satisfies the extensionality axioms \eqref{cext} and \eqref{fext2}, as well as the equational rule $\partial(f)=\Der f$, for $f$ a higher-order function (see \cite{LICS2021}, Lemma IV.1). 

%

We provide a sketch of the $\DTT$-category structure of the forgetful functor $U:\DLR\to \Set$ (given by $U(X,L,\rho)=X$ and $U(f,\varphi)=f$), described in detail in the Appendix. 
For any DLR $(X,L,\rho)$, a pure predicates $P\in\D P^{\flat}(X)$ is just a DLR $(X,L,\rho)$, with $(X\times X\mid P)=\coprod_{x,y\in X}\widetilde\rho(x,y) $ and projection $\pi_{X\times X}:(X\times X\mid P)\to X\times X$. 
For any set $Y$, $\D P(Y)$ is made of pullbacks $\langle f,g\rangle^{\sharp}Q$, where $Q=(X,L,\rho)\in \D P^{\flat}(X)$ and $f,g:Y\to X$, with 
$(Y\mid\langle f,g\rangle^{\sharp}Q)=\coprod_{y\in Y}\widetilde\rho(f(y),g(y))$ and associated projection
$\pi_{Y}:(Y\mid\langle f,g\rangle^{\sharp}Q)\to Y$.

For any separated DLR $(X,L,\rho)$, $\Dist X\in \D P^{\flat}(X)$ is $(X,L,\rho)$ itself, with $r_{X}(x)=\langle \langle x,x\rangle, \|x\|_{\rho}\rangle$;
moreover, for any binary predicate $P=\langle f,g\rangle^{\sharp}P^{\flat}\in \D P(X\times X)$ (with $(X\times X\mid P)=\coprod_{x,y\in X}\widetilde\mu(f(x,y),g(x,y))$) and function $c(x)=\langle \langle x,x\rangle, c'(x)\rangle: X\to (X\times X\mid P)$, so that $c' (x)\in \widetilde\mu(f(x,x),g(x,x))$, we can define a diagonal filler $j: (X\times X\mid \Dist X)\to  (X\times X\mid P)$ by 
\begin{center}
\adjustbox{scale=0.95,center}{
$
(j(\langle\langle x,y\rangle,\epsilon\rangle))_{2}=
\sup\{
c'(x), \inf\widetilde\mu(f(x,z),g(x,z))\mid z\in X \land 
\rho(x,\epsilon,z)
\}\in \widetilde \mu(f(x,y),g(x,y))
$}
\end{center}
\noindent The $\beta$-rule  
$j(\langle\langle x,x\rangle,\langle \| x\|_{\rho}\rangle)=c(x)$ follows from the fact that $\rho$ is separated. The $\eta$-rule holds only if $\rho$ is \emph{complete}, i.e.~for all $x\in X$ and $\epsilon \in L$, $\sup\{\inf\widetilde\rho(x,y)\mid \rho(x,\epsilon,y)\}=\epsilon$.

%

%


%
%

\begin{remark}
When $h:X\to Y$ interprets a program $t \in A\to B$, the 
predicate $D_{B}(fx,fy)$ corresponds to the pullback of $D_{B}(x,y)$ along $f(x,y)=h(x)$ and  $g(x,y)=h(y)$; moreover, since $\lambda x.\partial (fx)$ is the function $c'(x)=\| h(x)\|_{\mu}$,  $\Der t$ yields then the map $x,y,\epsilon \mapsto  
\sup\{\|h(x), h(z)\|_{\mu}\mid z\in X\land \rho(x,\epsilon,z)   \} $, as desired. 
\end{remark}
\begin{remark}
From Theorem \ref{thm:wfs} it follows that  given DLR $(X,L,\rho)$ and $(Y,M,\mu)$, any function $f:X\to Y$ factors as
$X\stackrel{i_{f}}{\to} \coprod_{x\in X, y\in Y}\widetilde\mu(f(x),y) \stackrel{\pi_{2}\circ \pi_{1}}{\to}Y$
where 
 $i_{f}(x)=\langle \langle x, f(x)\rangle, \|f(x)\|_{\mu}\rangle$.
\end{remark}

\begin{example}\label{example:3}
Suppose ${\widetilde{\TT H}}(\vec x)= (\sum_{i}x_{i})^{2}$, so that its derivative 
$\Der{{\widetilde{\TT H}}}$ is interpreted by the function 
$\varphi(\vec x, \vec \epsilon) = 2 (\sum x)(\sum\epsilon) +(\sum_\epsilon)^{2}$ (where $\sum y=|\sum_{i=0}^{N}y_{i}|$). Then for any function $f\in \Nat \to \Real$ with auxiliary map $\psi(n,\theta)$, the perforation error $b\in D_{\Real}(\TT H(f), \TT H(f^{*}))$ from Example \ref{example:1} corresponds to 
$\varphi([f(i), \psi(i,1)]_{i=0,\dots, N})$. 
\end{example}

\begin{example}
By interpreting $\Bool$ as the set $2=\{0,1\}$ and $D_{\Bool}$ as the DLR $(2,2,\rho_{2})$ corresponding to the discrete metric, the distance function $d_{A}\in (\Pi x,y\in A)D_{A}(x,y)$ from Example \ref{ex:distance} yields for any DLR $(X,L,\rho)$ the function $x,y \mapsto \| x,y\|_{\rho}: X\times X\to L$.
\end{example}

\section{Change Structures and Incremental Computation}\label{sec5}

%
%
%

%
%
%
%
%
%
%
%
%
%
%
%

In many situations in programming it happens that, after running a program $f$ on some input $i$, one needs to re-run $f$ over some slightly changed input $i'$; incremental computation is about finding ways to 
optimize this second computation without having to re-run $f$ from scratch on the novel input.
For example (we take this example from \cite{Cai2014}) suppose $f$ computes the sum of a finite bag of natural numbers $x$. Suppose $f$ has been run on $x=\{1,2,3,4\}$, and now needs to be re-run on $x'=\{2,3,4,5\}$; 
then we can compute $f(x')$ incrementally as follows: first, let the \emph{change} between $x$ and $x'$ be the pair of bags $\de x=(\{-1\}, \{+5\})$ describing what has to be changed to turn $x$ into $x$'. Then the change $\de y$ between $f(x)$ and $f(x')$ is the value $\de y=-1+5=+4$, and $f(x')$ can be computed simply by adding $\de y$ to $f(x)$.  
In particular, the operation $\de f(x, \de x)$ taking a bag and a change and returning the change $\de y$ is called (once more!) a \emph{derivative} of $f$ (we call it a \emph{change derivative} for clarity), as it describes the change to get from $f(x)$ to $f(y)$ as a function of $x$ and the change $\de x$.

In \cite{Cai2014} these ideas have been turned into a \emph{change semantics} for the simply typed $\lambda$-calculus, which was later generalized and simplified through the theory of \emph{change actions} \cite{Picallo2019, Picallo2019b}. These approaches have been applied to model different forms of discrete and automated differentiation \cite{Picallo2019, Pearlmutter2016}, and more recently related to models of the differential $\lambda$-calculus \cite{dallago2}. 
We focus here on change structures from \cite{Cai2014} since they have a natural higher-order structure. 

A \emph{change structure} (in short, CS) is a tuple $(X, \Delta_{X}, \oplus, \ominus)$ where $X$ is a set, for all $x\in X$, $\Delta_{X}x$ is a set of \emph{changes over $x$}, and $\oplus: X\times \Delta_{X}\to X$ and  $\ominus: X\times X\to \Delta_{X}$ are operations satisfying (1) $x\ominus y\in \Delta_{X}y$ and (2) $x\oplus(y\ominus x)=y$. 

Any function $f:X\to Y$ admits a change derivative $\de f: X\times \Delta X\to \Delta Y$ defined by
$\de f(x,\de x)=f(x\oplus \de x)\ominus f(x)$ and satisfying
$
f( x\oplus \de x)= f(x)\oplus \de f (x, \de x)
$. In other words, $\de f(x,\de x)$ describes the change needed to pass from $f(x)$ to $f(x\oplus \de x)$. 

For all $x\in X$ we let $\B 0_{x}= x\ominus x \in D_{X}x$; notice that $x\oplus \B 0_{x}=x$ and 
$ \de f(x,\B 0_{x})= \B 0_{f(x)}$. Moreover, whenever $x\oplus \de x = y$, we let $\ominus\de x:= x\ominus y$. 
Finally, given $\de x \in \Delta_{X}x$ and $\de y\in \Delta_{X}(x\oplus \de x)$, we let $\de x+\de y:=
((x\oplus \de x)\oplus \de y) \ominus x$. 

%
%


The CS semantics of $\DTT$ will interpret $D_{A}(\_,\_)$ as the type of changes over $A$, and a judgement $a\in D_{A}(t,u)$ as expressing the fact that $a$ is a change from $t$ to $u$ (i.e.~that $t\oplus a=u$). Self-differences $\partial(t)$ will correspond to the null change $\B 0_{t}$, and $\Der f$ will correspond to the change derivative $\de f$.
Moreover, due to the higher-order structure of change structures (recalled in the Appendix), this model satisfies the extensionality axiom \eqref{cext}. 



We provide a sketch of the $\DTT$-category structure of the forgetful functor $U:\CS \to \Set$, leaving as usual most details to the Appendix.
For any CS $X$, a pure predicate $P\in\D P^{\flat}(X)$ is a CS $P=(X,\Delta_{X},\oplus,\ominus)$, with 
$(X\times X\mid P)=\coprod_{x,y\in X}\Delta_{X}(x,y)$ (where $\Delta_{X}(x,y)=\{\de x \in \Delta_{X}x \mid x\oplus \de x=y\}$), and projection $\pi_{X}: (X\times X\mid P)\to X\times X$. 
For any CS $Y$, $\D P(Y)$ is made of pullbacks $\langle f,g\rangle^{\sharp}P$, where $P\in \D P^{\flat}(X)$ is a CS $P=(X,\Delta_{X},\oplus,\ominus)$, and 
$(Y\mid \langle f,g\rangle^{\sharp}P)=\coprod_{y\in Y}\Delta_{X}(f(y),g(y))$, with associated projection $\pi_{Y}:
(Y\mid \langle f,g\rangle^{\sharp}P)\to Y$.
%

%

For any CS $X$, $\Dist X\in \D P^{\flat}(X)$ is $X$ itself, with $r_{X}(x)=\langle \langle x,x\rangle, \B 0_{x}\rangle$; moreover,  for any binary predicate $P=\langle f,g\rangle^{\sharp}P^{\flat}\in \D P(X\times X)$ (with $(X\times X\mid P)
=\coprod_{x,y\in X}\Delta_{Y}(f(x,y),g(x,y))$) and function
$c: X\to (X\times X\mid P)$, where
 $c(x)=\langle \langle  x,x\rangle, c'(x)\rangle$, with $c'(x)\in \Delta_{Y}(f(x,x),g(x,x))$, a diagonal filler $j:
( X\times X\mid \Dist X) \to( X\times X\mid P)$ is defined by
\begin{center}
\adjustbox{scale=0.9}{$
\big(j(\langle \langle x,x'\rangle,\de x\rangle) \big)_{2}=  
\de f(\langle x,x\rangle, \langle \B 0_{x}, \ominus \de x\rangle)
+ c'(x) + 
\de g(\langle x,x\rangle, \langle \B 0_{x},\de x\rangle) \in \Delta_{Y}(f(x,y), g(x,y))
%
$}
\end{center}
One can check that $(j(\langle\langle x,x\rangle,\B 0_{x}\rangle))_{3}= c(x)$; the validity of the $\eta$-rule requires the further assumption that for all $x,y$, $y\ominus x$ is the \emph{unique} change from $x$ to $y$.

\begin{remark}
When $h:X\to Y$ interprets some program $f\in A\to B$, the interpretation of $\Der t$ corresponds to constructing a diagonal filler as above with $f(x,y)=h(x)$, $g(x,y)=h(y)$ and $c'(x)=\B 0_{h(x)}$; then one obtains the map $x,y,\de x \mapsto \de h(x,\de x)$ as desired. 
\end{remark}


%
%
\begin{remark}
From Theorem \ref{thm:wfs} it follows that  given change structures on $X,Y$, any function $f:X\to Y$ factors as
$X\stackrel{i_{f}}{\to} \coprod_{x\in X, y\in Y}\Delta_{Y}(f(x),y) \stackrel{\pi_{2}\circ \pi_{1}}{\to}Y$
where 
 $i_{f}(x)=\langle \langle x, f(x)\rangle, \B 0_{f(x)}\rangle$.
\end{remark}
%

\begin{example}
Consider the change structure on $\Real$ where $\Delta_{\Real}x= \BB R$, $\oplus$ is addition and $\ominus$ is subtraction. Let $f,g:\Nat\to \Real$ be such that $g$ ``increments'' on $f$ through some function $F: \Nat\times \Nat \to \Real$ (i.e.~$g(x)=f(x)+F(x,0)$). In $\DTT$ $F$ yields then an element of $D_{\Nat\to \Real}(f,g)$. By reasoning as in Example \ref{example:1} we can construct the increment
from $\TT  H(f)$ to $\TT H(g)$; if $\TT H(f)$ is as in Example \ref{example:2}, this corresponds then to
$1/N\cdot \sum_{i=0}^{N}F(i,0)$. 
\end{example}

\section{Cartesian Differential Categories}\label{sec6}

%

We conclude our sketch of differential models of $\DTT$ with the axiomatization of 
program derivatives provided by cartesian differential categories (CDC). 
The introduction of the \emph{differential $\lambda$-calculus} \cite{difflambda}, an extension of the $\lambda$-calculus with a \emph{differential} operator $\de t$, has motivated much research on abstract axiomatizations of differentiation that generalize the usual derivatives from calculus to higher-order programming languages \cite{Blute2009, Manzo2010, Blute2019}. 
CDC can be seen as a common ground for all these approaches, as they provide basic algebraic rules for derivatives in a cartesian setting.  

We recall that a cartesian category  $\BB C$ is \emph{left-additive} when the Hom-objects of $\BB C$ are monoids, with the monoidal operations $0,+$ commuting with the cartesian structure (e.g.~$f+(g\times h)=(f+g)\times (f+h)$), and satisfies left-additivity, i.e.~$0\circ f=0$ and $(g+h)\circ f=(g\circ f)+(h\circ f)$. 
A CDC is a cartesian left-additive category endowed with a {derivative operator} $\de$ such that for all $f: X\to Y$, $\de f: X\times X\to Y$, satisfying a few axioms (D1)-(D7) (recalled in the Appendix). Intuitively, $\de f(x,y)$ describes the differential $f'(x) \cdot y$ of $f$ at $x$, so it should be a \emph{linear} function in $y$. This is reflected by the axiom (D2) stating, informally, that $\de f$ is additive in its second variable, i.e.~$\de f(x,0)=0$ and $\de f(x,y+y')=\de f(x,y)+\de f(x,y')$. 
Among the other axioms for $\de f$ we find analogs of axioms \eqref{dchain} (the chain rule) as well as axioms expressing the commutation of $\de f$ with the cartesian structure, plus some other axioms concerning second derivatives.
When a CDC is cartesian closed, one usually adds also axiom \eqref{jl1} (called $D$-Curry in \cite{Manzo2010}), and one speaks of a \emph{differential $\lambda$-category}.

We now describe the interpretation of $\DTT$ in a differential $\lambda$-category. In fact, everything works in any CDC if we forget about the higher-order structure.\footnote{Since Euclidean spaces $\BB R^{n}$ and smooth functions form a CDC \cite{Blute2009} this shows in particular that one can consistently interpret $\Der f$ as the ``true'' derivative from analysis.} We will interpret  $D_{A}(t,u)$ as a sort of ``tangent space'' of $t$ (notice that we  ignore $u$); the self-difference $\partial(t)$ will correspond to the zero vector $0$, and the derivative $\Der f$, which sends ``vectors tangent to $x$'' into ``vectors tangent to $fx$'', will correspond to $\de f$. 
Moreover, the resulting model satisfies the extensionality axioms \eqref{cext} and \eqref{fext1} (see the Appendix for details).

Let $\BB C$ be a differential $\lambda$-category. We sketch the $\DTT$-structure on the identity functor $\mathrm{id}:\BB C\to \BB C$, leaving all details to the Appendix. The classes $\D P^{\flat}(X)$ contain all objects of the form $X^{n}$ (where $X^{0}=1$, $X^{n+1}=X^{n}\times X$), with 
$(X^{2}\mid X^{n})= X^{2}\times X^{n}$ and projection $\pi_{X}:X^{2}\times X^{n}\to X^{2}$. 
$\D P(Y)$ is made of pullbacks $\langle f,g\rangle^{\sharp}X^{n}$, for arrows $f,g:Y\to X$, with 
$(Y\mid \langle f,g\rangle^{\sharp}X^{n})= Y\times X^{n}$ and projection $\pi_{Y}:(Y\mid \langle f,g\rangle^{\sharp}X^{n})\to Y$.
%
%

For any object $X$, the pure predicate $\Dist X\in \D P^{\flat}(X)$ is just $X$, and $r_{X}:X\to X^{3}$ is $r_{X}=\langle \langle\mathrm{id_{X}}, \mathrm{id}_{X}\rangle,0\rangle$.
%
%
%
Moreover, for any binary predicate $P=\langle f,g\rangle^{\sharp}Z^{n}\in \D P(X^{2})$, where $f,g:X^{2}\to Z$, for any arrow 
$c: X\to (X^{2}\mid P)=X^{2}\times Z^{n}$, where 
$c=\langle \mathrm{id}_{X},\mathrm{id}_{X},c'\rangle$, with $c': X\to Z^{n}$, we can define a diagonal filler $j: X^{3} \to X^{2}\times Z^{n} $ by letting
\begin{center}
$
j = \langle \langle\pi_{1}\circ \pi_{1}, \pi_{2}\circ \pi_{1}\rangle, c'\circ( \pi_{1}\circ \pi_{1})+ \langle\de f\circ \langle\pi_{1}\circ \pi_{1},\pi_{2} \rangle\rangle^{n}\rangle$
\end{center}
Observe that $(j\circ r_{X})_{2}=c$ and moreover, if $Z^{n}=X$, $f=\mathrm{id}_{X}$ and $c=r_{X}$, 
then $j= \mathrm{id}_{X\times X}$, so both the $\beta$- and $\eta$-rules are satisfied.

\begin{remark}
When $h:X\to Y$ interprets $t\in A\to B$, $\Der h$  corresponds to a diagonal filler with $f=h\circ \pi_{1},g=\pi_{2}$ and $c=\langle \mathrm{id}_{X}, 0\rangle$, yielding the map $ \de h:X\times X\to Y$.
\end{remark}

\begin{remark}
From Theorem \ref{thm:wfs} it follows that the classes $\C L_{\D P}, \C R_{\D P}$ form a WFS in $\BB C$, where $f:X\to Y$ factors as $X\stackrel{i_{f}}{\to} (X\times Y)\times Y \stackrel{\pi_{2}\circ \pi_{1}}{\to}X$, with 
$i_{f}=\langle\langle\mathrm{id}_{X},f\rangle,0\rangle$.
\end{remark}

\section{Conclusions}\label{sec8}

\subparagraph*{Related Work}

$\DTT$ is definitely not the first proof system proposed to formalize relational reasoning for higher-order programs (nor the first one based on dependent types, e.g.~\cite{Stewart2013}). 
Among the many existing systems we can mention the logics for parametricity and logical relations \cite{Plotkin1993, IRIS}, the refinement type systems for cryptography \cite{Barthe2014}, differential privacy \cite{Barthe_2012} and relational cost analysis \cite{Barthe2017}, Relational Hoare Type Theory \cite{Stewart2013} and Relational Higher-Order Logic \cite{Barthe2017}. 
In particular, it is tempting to look at $\DTT$ as a 
\emph{proof-relevant} variant of (some fragment of) RHOL, since 
the latter is based on judgements of the form $\Gamma\mid \Psi\vdash t:A\sim u:B \mid \varphi$, where $\Gamma\vdash t:A, u:B$ are typings in ST$\lambda$C, and $\Psi,\varphi$ are logical formulas depending on the variables in $\Gamma$ as well as $t$ and $u$. 
Indeed, the main difference between $\DTT$ and such systems is that program differences are represented as proof objects. This looks a rather natural choice at least for those semantics (like e.g.~DLR and CS) where program differences can be seen as being themselves some kind of programs.

Neither we are the first to observe formal correspondences between various notions of program derivative. For instance, a formalization of change structures in the context of DLR is discussed in \cite{dallago2}; 
connections between metric semantics and DLR are studied in \cite{Geoffroy2020, LICS2021}, based on  \emph{generalized} metric spaces and \emph{quantaloid}-enriched categories \cite{Stubbe2014}. In particular, the 
 DLR model sketched in Section \ref{sec4} can be used to provide a ``quantaloid-interpretation'' of $D_{A}(\_,\_)$ (to be compared with the groupoid structure of the identity type in full ITT). 
Recently, \emph{cartesian difference categories} \cite{Picallo2020} have been proposed as a general framework for program derivatives (unifying cartesian differential categories with approaches 
related to both discrete derivatives and 
 incremental computation). It seems that our model in Section \ref{sec4} can be extended to such categories in a straightforward way.

\subparagraph*{Future Work}

The main goal of this paper was to provide evidence that ITT could serve the purpose of formalizing differential reasoning. Yet, examples were left necessarily sketchy  
and more 
substantial formalization work (as well as implementations) needs to be addressed.

As $\DTT$ is a fragment of ITT, syntactic results like strong normalization follow. Yet, the problem should be addressed whether such results are stable also when further equations for derivatives (as those described in Section \ref{sec2}) are added. Moreover, it is well-known that a suitable formulation of 2-dimensional ITT satisfies a \emph{canonicity} condition \cite{Licata2012}: a closed normal term of type $\Bool$ is either $\B 0,\B 1$; it would be interesting to see whether this result can be scaled to the fragment $\DTT$.

Finally, it seems worth exploring extensions of $\DTT$ with further structure, for instance with dependent types at the base level (e.g.~following work on dependent types for differential privacy \cite{Gaboardi_2013}), as well as with primitives for probabilistic reasoning (as in \cite{Reed_2010, Barthe_2012, Barthe2014}).

%
%
%
%
%
%
%
%
%

\bibliography{main.bib}

\appendix

\section{Type Systems: Details}

The typing rules of $\DTT$ can be divided into the rules for typing program terms and the rules for typing difference terms.

\subparagraph*{Simple Types}
In its basic formulation, the rules for typing program terms are the standard typing rules of ST$\lambda$C, recalled in Fig.~\ref{fig:stlc}. 
In Section 4 we considered a variant of $\DTT$ with affine simple types and an exponential $!_{r}A$, for all $r\in \BB R_{\geq 0}$, corresponding to a fragment of $\mathrm{Fuzz}$ \cite{Reed_2010}. 
We describe the rules of this fragment, that we call ST$\lambda$C$^{!}$, in Fig.~\ref{fig:fuzz}.
The terms are generated by the grammar:
$$
t,u:= x\mid \lambda x.t\mid tu \mid !t \mid \LET !x=t\IN u\mid (t,u)\mid  \LET (x,y)=t\IN u 
$$
For the purposes of this article we limited ourselves to a minimal fragment of this language. For a more practical language see \cite{Reed_2010, Gaboardi2017}.  
Simple types are generated by the grammar below:
$$
A,B:= X\mid !_{r}A \quad (r\in \BB R_{\geq 0}) \mid A\multimap B \mid A\otimes B
$$
Type judgements are of the form $\Phi \vdash t:A$, where a context $\Phi$ is a list of declarations of the form $x\in_{r}A$, for some $r\in \BB R_{\geq 0}$.
We define the following operation $\Phi+\Psi$  as follows:
\begin{align*}
() + () & =() \\
(\Phi, x\in_{r} A)+( \Psi, x\in_{s} A) & =  (\Phi+\Psi), x\in_{r+s}A \\
(\Phi, x\in_{r}A)+\Psi & =(\Phi+\Psi), x\in_{r}A \qquad (x\notin \Psi) \\
\Phi+ (\Psi, x\in_{r}A) &= (\Phi+\Psi), x\in_{r} A \qquad (x\notin \Phi)
\end{align*}
Moreover, we let $s\Phi$ be the context made all judgmenets $x\in_{sr}A$, where $(x\in_{r}A)\in \Phi$.  

Observe that one can always type an affine term like e.g.~$\lambda xy.x$ with a linear type $A\multimap B\multimap A$. Instead, a term like $\lambda xy.x(xy)$ containing two occurrences of $x$ cannot be given the linear type $(A\multimap A)\multimap (A\multimap A)$ but a type of the form
$!_{2}(A\multimap A)\multimap (A\multimap A)$. 

There exists a ``forgetful'' translation $^{*}$ from ST$\lambda$C$^{!}$ to ST$\lambda$C given on terms by 
\begin{align*}
 & x^{*}=x \qquad (\lambda x.t)^{*}=\lambda x.t^{*} \qquad (tu)^{*}=t^{*}u^{*}\\
&  (!t)^{*}=t^{*} \qquad
(\LET !x =t\IN u)^{*}=(\lambda x.t^{*})u^{*} \\ 
& (t,u)^{*}=\langle t^{*},u^{*}\rangle \qquad (\LET (x,y)=t\IN u)^{*}= (\lambda xy.u^{*})\pi_{1}(t^{*})\pi_{2}(t^{*})
\end{align*}
and on types by
\begin{align*}
X^{*}=X \qquad (!_{r}A)^{*}=A^{*} \qquad (A\multimap B)^{*}=A^{*}\to B^{*} \qquad (A\otimes B)^{*}=A^{*}\times B^{*}
\end{align*}
This translation can be used to define the functor $H$ from Section 3 from the context category $\CTX^{!}_{0}$ of ST$\lambda$C$^{!}$ to $\CTX$.

\begin{figure}
\fbox{
\begin{minipage}{0.9\textwidth}
\begin{center}
\AXC{$x\in A\in \Phi$}
\UIC{$\Phi\vdash x\in A$}
\DP

\bigskip

\begin{tabular}{c c }
\AXC{$\Phi, x\in A\vdash t\in B$}
\UIC{$\Phi\vdash \lambda x.t\in A\to B$}
\DP

& 

\AXC{$\Phi \vdash t\in A\to B$}
\AXC{$\Phi\vdash u\in A$}
\BIC{$\Phi \vdash tu\in B$}
\DP
\\ 

& 

\\

\AXC{$\Phi \vdash t\in A$}
\AXC{$\Phi \vdash u\in B$}
\BIC{$\Phi \vdash \langle t,u\rangle \in A\times B$}
\DP

& 

\AXC{$\Phi \vdash t\in A_{1}\times _{2}$}
\UIC{$\Phi \vdash \pi_{i}(t)\in A_{i}$}
\DP

\end{tabular}
\end{center}
\end{minipage}
}
\caption{Typing rules for ST$\lambda$C.}
\label{fig:stlc}
\end{figure}

\begin{figure}
\fbox{
\begin{minipage}{0.9\textwidth}
\begin{center}
\AXC{$x\in_{r} A\in \Phi$}
\UIC{$\Phi\vdash x\in A$}
\DP

\bigskip

\begin{tabular}{c c }
\AXC{$\Phi, x\in_{1} A\vdash t\in B$}
\UIC{$\Phi\vdash \lambda x.t\in A\multimap B$}
\DP

& 

\AXC{$\Phi \vdash t\in A\multimap B$}
\AXC{$\Phi\vdash u\in A$}
\BIC{$\Phi \vdash tu\in B$}
\DP
\\ 

& 

\\

\AXC{$\Phi \vdash t\in A$}
\AXC{$\Psi \vdash u\in B$}
\BIC{$\Phi+\Psi \vdash ( t,u) \in A\otimes B$}
\DP

& 

\AXC{$\Phi\vdash t\in A \otimes B$}
\AXC{$\Psi,x\in_{r}A, y\in_{r} B\vdash u:C$}
\BIC{$\Phi+\Psi \vdash \LET (x,y)=t \IN u\in C$}
\DP
\\
& 

\\

\AXC{$\Phi \vdash t\in A$}
\UIC{$s\Phi \vdash !t \in \ !_{s}A$}
\DP

& 

\AXC{$\Phi\vdash t\in !_{s}A$}
\AXC{$\Psi,x\in_{rs}A\vdash u:C$}
\BIC{$r\Phi+\Psi \vdash \LET !x=t \IN u\in C$}
\DP

\end{tabular}
\end{center}
\end{minipage}
}
\caption{Typing rules for ST$\lambda$C$^{!}$.}
\label{fig:fuzz}
\end{figure}

\subparagraph*{Difference Types}

As discussed in Section 2, the rules for difference terms are the standard rules of ITT, restricted to the language of $\DTT$.
We illustrate in Fig.~\ref{fig:diffrules1} the rules for the difference type (where, compared to the rules sketched in Section 2, we highlight the role of contexts), and all other rules in Fig.~\ref{fig:diffrules2}.
Finally, we illustrate $\beta$- and $\eta$-rules in Fig.~\ref{fig:betaeta}.

\begin{figure}[t]
\fbox{
\begin{minipage}{0.9\textwidth}
\begin{center}
$\AXC{$\B x\in \Phi_{0} \vdash t\in A$}
\UIC{$(\B x\in \Phi_{0}\mid\BS \epsilon \in \Phi_{1}(\B x))\vdash \partial(t)\in D_{A}(t,t)$}
\DP$

\bigskip

$
\AXC{$(\B x\in \Phi_{0}\mid\BS \epsilon \in \Phi_{1}(\B x,t,u))\vdash a\in D_{A}(t,u)$}
\noLine
\UIC{$(\B x\in \Phi_{0},x,y\in A\mid\BS \epsilon \in \Phi_{1}(\B x,x,y))\vdash \C C(\B x,x,y)\in \DType$}
\noLine
\UIC{$(\B x\in \Phi_{0},x\in A\mid\BS \epsilon \in \Phi_{1}(\B x,x,x))\vdash b \in \C C(\B x,x,x)$}
\UIC{$(\B x\in \Phi_{0}\mid\BS \epsilon \in \Phi_{1}(\B x,t,u))\vdash
\J (t,u,a,[x]b)\in \C C(\B x, t,u)$}
\DP  
$
\end{center}
\end{minipage}
}
\caption{Typing rules for the difference type.}
\label{fig:diffrules1}
\end{figure}

\begin{figure}
\fbox{
\begin{minipage}{0.9\textwidth}
\begin{center}
\AXC{$(\B x\in \Phi_{0}, x\in A\mid\BS \epsilon \in \Phi_{1}(\B x))\vdash a\in \C A(\B x, x)$}
\UIC{$(\B x\in \Phi_{0}\mid\BS \epsilon \in \Phi_{1}(\B x))\vdash\lambda x. a\in(\Pi x\in A) \C A(\B x, x)$}
\DP

\bigskip

\AXC{$(\B x\in \Phi_{0}\mid\BS \epsilon \in \Phi_{1}(\B x))\vdash a\in (\Pi x\in A)\C A(\B x, x)$}
\noLine
\UIC{$(\B x\in \Phi_{0})\vdash t \in A$}
\UIC{$(\B x\in \Phi_{0}\mid\BS \epsilon \in \Phi_{1}(\B x))\vdash at\in \C A(\B x, t)$}
\DP

\bigskip

\AXC{$(\B x\in \Phi_{0}, x,y\in A\mid\BS \epsilon \in \Phi_{1}(\B x), \epsilon\in D_{A}(x,y)\vdash a\in \C A(\B x, x,y)$}
\UIC{$(\B x\in \Phi_{0}\mid\BS \epsilon \in \Phi_{1}(\B x))\vdash\lambda xy\epsilon. a\in(\Pi x,y\in A)(D_{A}(x,y)\to \C A(\B x, x,y))$}
\DP

\bigskip

\AXC{$(\B x\in \Phi_{0}\mid\BS \epsilon \in \Phi_{1}(\B x))\vdash a\in (\Pi x,y\in A)(D_{A}(x,y)\to\C A(\B x, x,y))$}
\noLine
\UIC{$(\B x\in \Phi_{0}\mid\BS \epsilon \in \Phi_{1}(\B x))\vdash b\in D_{A}(t,u)$}
\UIC{$(\B x\in \Phi_{0}\mid\BS \epsilon \in \Phi_{1}(\B x))\vdash atub\in \C A(\B x, t,u)$}
\DP

\bigskip

\AXC{$(\B x\in \Phi_{0}\mid\BS \epsilon \in \Phi_{1}(\B x))\vdash a\in \C A$}
\AXC{$(\B x\in \Phi_{0}\mid\BS \epsilon \in \Phi_{1}(\B x))\vdash b\in \C A$}
\BIC{$(\B x\in \Phi_{0}\mid\BS \epsilon \in \Phi_{1}(\B x))\vdash \langle a,b\rangle\in \C A\times \C B$}
\DP

\bigskip

\AXC{$(\B x\in \Phi_{0}\mid\BS \epsilon \in \Phi_{1}(\B x))\vdash a\in \C A_{1}\times \C A_{2}$}
\UIC{$(\B x\in \Phi_{0}\mid\BS \epsilon \in \Phi_{1}(\B x))\vdash \pi_{i}(a)\in \C A_{i}$}
\DP

\end{center}
\end{minipage}
}
\caption{Typing rules of $\DTT$.}
\label{fig:diffrules2}
\end{figure}

\begin{figure}
\fbox{
\begin{minipage}{0.9\textwidth}
\begin{align*}
(\lambda x.a)t & \simeq_{\beta} a[t/x]  &  (\lambda x.ax)&\simeq_{\eta} a \quad (x\notin \mathrm{FV}(a))\\
(\lambda xy\epsilon.a)tub & \simeq_{\beta} a[t/x,u/y,b/\epsilon) & 
(\lambda xy\epsilon.axy\epsilon)& \simeq_{\eta}a \quad (x,y,\epsilon\notin \mathrm{FV}(a))\\
\pi_{i}(\langle a_{1},a_{2}\rangle)& \simeq_{\beta} a_{i} &
\langle \pi_{1}(a), \pi_{2}(a)\rangle &\simeq_{\eta} a \\
\J (t,t,\partial(t),[x]b) & \simeq_{\beta} b[t/x] &
\ J(t,u,a,[x]\partial(x)) & \simeq_{\eta} a
\end{align*}
\medskip
\end{minipage}
}
\caption{$\beta$- and $\eta$-rules for difference terms.}
\label{fig:betaeta}
\end{figure}

\begin{proof}{Proof of Lemma \ref{lemma:pure}}
\begin{itemize}
\item if $\C C(\B z)=D_{C}(t, u)$, then $(t,u):(\B z\in \Phi_{0})\to (y,y'\in C)$ we let $\Psi_{0}=C$ and $\C C^{\flat}(y,y')=D_{C}(y,y')$.

\item if $\C C(\B z)= \C C_{1}(\B z)\times \C C_{2}(\B z)$, then by induction hypothesis there exist pure predicates
$( x,  y\in A_{1})\B C_{1}^{\flat},(x,y\in A_{2})\B C_{2}^{\flat}$,  and terms $(t_{1},u_{1}):(\B z\in \Phi_{0})\to (x,y\in A_{1})$ and
$(t_{2},u_{2}):(\B z\in \Phi_{0})\to (x,y\in A_{2})$ such that $\C C_{1}(\B z)=\C C_{1}^{\flat}( t_{1}, u_{1})$ and  $\C C_{2}(\B z)=\C C_{2}^{\flat}( t_{2}, u_{2})$. We can let then
$A= A_{1}\times A_{2}$, $ t=\langle  t_{1},\B t_{2}\rangle$, $u=\langle u_{1}, u_{2}\rangle$ and 
$\C C^{\flat}(w,w')= \C C_{1}^{\flat}(\pi_{1}(w), \pi'(w'))\times \C C_{2}^{\flat}(\pi_{2}(w), \pi_{2}(w'))$.

\item if $\C C(\B z)=(\Pi w\in D)\C B(\B z,w)$ then by induction hypothesis there exists a pure predicate
$(x,y\in A')\C B^{\flat}(x,  y)$ and terms $( t', u'):(\B z\in \Phi_{0}, w\in D)\to ( x, y\in A')$ such that $\C B(\B z,w)=\C B^{\flat}( t',  u')$. We let then $A= D\to A'$, $\C C^{\flat}( x,  y)=(\Pi w\in A)\C B^{\flat}(  xw,  yw )$ and $t=\lambda w.t'$, $u=\lambda w.u'$.

\item if $\C C(\B z)=(\Pi w,w'\in C)(D_{C}(w,w')\to \C B(\B z,w,w')$ then by induction hypothesis there exists a pure predicate $( x,  y\in A')\C B^{\flat}( x,  y)$ and terms $t',u':(\B z\in \Phi_{0}, w,w'\in A) \to ( x,  y\in A')$ such that $\C B(\B z, w,w')=\C B^{\flat}( t',  u')$. 
We let then  $A= C\to (C\to A')$,  $\C C^{\flat}(x,  y)=(\Pi w,w'\in C)(D_{A}(w,w')\to\C B^{\flat}( xww',  yww' ))$  and $t=\lambda ww'. t'$, $u= \lambda ww'.u'$.
%
%
\end{itemize}
\end{proof}

\section{Equational Rules for Derivatives: Details}

We list a few equational rules for the operators $\J$ and $\mathsf D$, that make sense under the validity of some of the extensionality axioms.

\begin{itemize}
\item in presence of \eqref{cext} one can consider the following rules:
\begin{align*}
\partial(\langle t,u\rangle) & = \langle\partial(t), \partial(u)\rangle \tag{$\J\times a$}\label{jca}\\
\J (t,u,a, [x]\langle b,c\rangle)&=\langle \J (t,u,a,[x],b),\J(t,u,a,[x],b)\rangle \tag{$\J\times b$} \label{jcb}\\
\J (t,u,a,[x]\pi_{i}(b)) & = \pi_{i}\big ( \J (t,u,a,[x]b)\big)\tag{$\J\times c$}\label{jc}
\end{align*}
These rules say that the difference structure commutes with the cartesian structure. 
\item in presence of \eqref{fext1} one can consider the following rules:
\begin{align*}
\partial(\lambda x.t) & = \lambda x.\partial(t) \tag{$\J \lambda 1 a$}\label{jl1a}\\
\J (t,u,a,[x]\lambda y.b(x,y)) & =\lambda y. \J (\langle t,y\rangle,\langle u,y\rangle,\langle a,\partial(y)\rangle,[z]b(\pi_{1}(z),\pi_{2}(z))
\tag{$\J\lambda 1b$}\label{jl1} 
\end{align*}
from which it follows that the derivative ``in $x$''  $\Der{\lambda xy.f(x,y)}$ of some binary function $f(x,y)$ is the same as the derivative ``in $z=\langle x,y\rangle$'', i.e.~$\Der{\lambda z.f(\pi_{1}(z), \pi_{2}(z))} $, where the error on $y$ is its self-difference, i.e.~$\lambda xx'\delta y.\Der{\lambda z.f(\pi_{1}(z), \pi_{2}(z))
}\langle x,y\rangle \langle x',y\rangle\langle \epsilon, \partial(y)\rangle$ (when $\partial(y)$ is interpreted as the null error $0$, this coincides with axiom $D$-curry from \cite{Manzo2010});
%

\item in presence of \eqref{fext2} and \eqref{cext} one can consider the following rules:
\begin{align*}
\partial(\lambda x.t) & = \Der{\lambda x.t} \tag{$\J\lambda 2a$}\label{jl2a}\\
\J (t,u,a, [x]\lambda y.b)  & =
\lambda yy'\delta.\J\Big (\langle t,y\rangle, \langle u,y'\rangle, \langle a,\delta\rangle, [z]\big (b[\pi_{1}(z)/x, \pi_{2}(z)/y]\big)\Big)
\tag{$\J\lambda 2b$}\label{jl2} 
%
\end{align*}
This shows that the axioms \eqref{fext1} and \eqref{fext2} might give rise to derivative with a rather different operational semantics. In particular, Eq.~\eqref{jl2a} says that the self-difference of a function coincides with its derivative (this will be the case in the models from Sec.~\ref{sec4} and \ref{sec5}); moreover, \eqref{jl2} says that the derivative ``in $x$'' of a binary function $f(x,y)$ actually derives also in the variable $y$. In particular one can deduce that $\Der{\lambda xy.t(x,y)}  xx'\epsilon yy'\delta$ is the same as $\Der{\lambda z.t(\pi_{1}(z),\pi_{2}(z))}\langle x,y\rangle\langle x'y'\rangle \langle \epsilon,\delta\rangle$.
\end{itemize}

%
%
%


\section{$\DTT$-Categories: Details}

In this section we define in detail the notion of $\DTT$-category that was sketched in Section \ref{sec3}. 
Throughout this section we suppose $U:\BB C_{0}\to \BB C$ to be a strict monoidal functor, where $\BB C_{0}$ is a symmetric monoidal category and $\BB C$ is a cartesian category.
Moreover, we suppose that for any object $\Gamma$ of $\BB C_{0}$ the following data is given:
\begin{itemize}
\item a collection $\D P(\Gamma)$ of predicates over $\Gamma$, and for each $P\in \D P(\Gamma)$ an object $\Gamma\mid P$ and an arrow $\pi_{\Gamma}: \Gamma\mid P\to U\Gamma$; we further require that:
	\begin{itemize}
	\item  for all $P\in \D P(\Gamma)$ and $f:U\Delta \to U\Gamma$, the pullback $f^{\sharp}P$ exists and is in $ \D P(\Delta)$;	
	\item when the monoidal product of $\BB C_{0}$ is not cartesian, we furthermore require that:
		\begin{itemize}
		 \item for all $P\in \D P(\Gamma)$, the pullback $\pi_{U\Gamma}^{\sharp}P$ exists and is in $ \D P(\Delta.\Gamma)$, where $\pi_{U\Gamma}:U\Delta\times U\Gamma \to U\Gamma$;
		  \item for all $P\in \D P(\Gamma.\Gamma)$, the pullback $\delta_{U\Gamma}^{\sharp}P$ exists and is in $ \D P(\Gamma)$, where $\delta_{U\Gamma}:U\Gamma \to U\Gamma\times U\Gamma$;
		\end{itemize}

	\item we require all mentioned pullbacks to be associative and unital;
	\item for all $P,Q\in \D P(\Gamma)$, a predicate $P\times Q\in \D P(\Gamma)$ exists such that $\Gamma\mid P\times Q$ is the cartesian product of $\Gamma\mid P$ and $\Gamma\mid Q$ in the \emph{slice category} $\BB C_{\D P}^{\Gamma}$, i.e.~in the category of predicates in $\D P(\Gamma)$ and \emph{vertical} arrows (see Section 3), i.e.~those arrows $h: \Gamma\mid P\to \Gamma\mid Q$ making the diagram below commute
	$$
	\begin{tikzcd}
	\Gamma\mid P\ar{d}[left]{\pi_{\Gamma}} \ar{r}{h} & \Gamma\mid Q \ar{d}{\pi_{\Gamma}} \\
	U\Gamma\ar[-, double]{r} & U\Gamma
	\end{tikzcd}
	$$
	
		\end{itemize}

\item a sub-collection $\D P^{\flat}(\Gamma)\subseteq \D P(\Gamma.\Gamma)$ of \emph{pure predicates} 
generating  $\D P$, (i.e.~such that any $P\in \D P(\Gamma)$ is of the form $ f^{\sharp}P^{\flat}$, where $P^{\flat}\in \D P^{\flat}(\Delta)$ and $f\in \BB C_{0}(\Delta, \Gamma.\Gamma)$), 
together with a chosen pure predicate $\Dist \Gamma\in \D P^{\flat}(\Gamma)$, and closed with respect to the following conditions: 

\begin{itemize}

\item for all pure predicates $P,Q\in \D P^{\flat}(\Gamma)$,  $P\times Q\in \D P^{\flat}(\Gamma)$;
%
\item for all objects $\Delta, \Gamma$ of $\BB C_{0}$ and pure predicate $P\in \D P^{\flat}(\Gamma)$, a pure predicate $\Pi _{\Delta}P\in \D P^{\flat}(\Gamma^{\Delta}.\Gamma^{\Delta})$ such that for all $Q\in \D P(\Sigma)$ there is a bijection
$$
\BB C\big ( \Delta. \Sigma\mid\pi_{\Sigma}^{\sharp}Q \ ,\Gamma.\Gamma\mid P\big ) \simeq
\BB C\big ( \Sigma\mid Q \ , \  \Gamma^{\Delta}.\Gamma^{\Delta}\mid \Pi_{\Delta}P\big )
$$
where $\pi_{\Sigma}$ is the projection $\Delta. \Sigma\to \Sigma$.
%
%
Moreover, we require that dependent products commute with pullbacks, i.e.~for all $f\in \BB C_{0}(\Gamma', \Gamma)$, $f^{\sharp}(\Pi_{\Delta}P)= \Pi_{\Delta}(f^{\sharp}P)$.


\item for all objects $\Delta,\Gamma$ of $\BB C_{0}$ and pure predicate $P\in \D P^{\flat}(\Gamma.\Gamma)$  a pure predicate
$P^{\left(\Dist \Delta\right)}\in \D P^{\flat}(\Gamma^{\Delta.\Delta})$ such that for all $Q\in \D P(\Sigma)$ there is a bijection
$$
\BB C\big ( \Delta. \Delta.\Sigma\mid  \pi_{\Sigma}^{\sharp}Q  \times\Dist \Delta   \ , \   \Gamma.\Gamma\mid P \big)
\simeq
\BB C\big ( \Sigma\mid Q \ ,\   \Gamma^{\Delta.\Delta}. \Gamma^{\Delta.\Delta}\mid P^{\left(\Dist \Delta\right)}  \big)
$$
where $\pi_{\Sigma}$ is the projection $\Delta. \Delta .\Sigma\to \Sigma$.
%
%
%
Moreover, 
we require that 
 for all $f\in \BB C_{0}(\Gamma', \Gamma)$, $f^{\sharp}\left(P^{\left(\Dist{\Delta}\right)}\right)= (f^{\sharp}P)^{\left(\Dist{\Delta}\right)}$.
 
\end{itemize}

\end{itemize}
%
%

We will make extensive use of the following constructions: for all objects $\Gamma, \Delta $ of $\BB C_{0}$ there exists 
\begin{itemize}
\item a functor $\pi_{\Delta.\Gamma}^{\sharp}: \BB C_{\D P}^{\Delta.\Gamma}\to \BB C_{\D P}^{\Delta.\Gamma.\Gamma}$ induced by the projection $\pi_{\Delta.\Gamma}: \Delta.\Gamma.\Gamma\to \Delta.\Gamma$ which deletes the third component of $\Delta.\Gamma.\Gamma$, together with an arrow
$\pi_{\Delta.\Gamma}^{+}:\Delta.\Gamma.\Gamma\mid \pi_{\Delta.\Gamma}^{\sharp}(P)\to  \Delta.\Gamma\mid P$ making the pullback diagram below commute:
$$
\begin{tikzcd}
\Delta.\Gamma.\Gamma\mid \pi_{\Delta.\Gamma}^{\sharp}(P)\ar{d}[left]{\pi_{\Delta.\Gamma.\Gamma}} \ar{r}{\pi_{\Delta.\Gamma}^{+}} & \Delta.\Gamma\mid P \ar{d}{\pi_{\Delta.\Gamma}}  \\
U\Delta.U\Gamma.U\Gamma \ar{r}{\pi_{\Delta.\Gamma}} & U\Delta.U\Gamma
\end{tikzcd}
$$

\item a functor $\delta_{\Delta.\Gamma\mid Q}^{\sharp}: \BB C_{\D P}^{\Delta.\Gamma.\Gamma}\to \BB C_{\D P}^{\Delta.\Gamma}$ induced by the arrow $\delta_{\Delta.\Gamma}: \Delta.\Gamma\to \Delta.\Gamma.\Gamma$ which duplicates the second component of $\Delta.\Gamma$, together with an arrow
$\delta_{\Delta.\Gamma\mid Q}^{+}: \Delta.\Gamma\mid\delta_{\Delta.\Gamma\mid Q}^{\sharp}(P) \to \Delta.\Gamma\mid Q\times P$ making the pullback diagram below commute:
$$
\begin{tikzcd}
\Delta.\Gamma\mid \delta_{\Delta.\Gamma}^{\sharp}(P)\ar{d}[left]{\pi_{\Delta.\Gamma}} \ar{r}{\delta_{\Delta.\Gamma}^{+}} & \Delta.\Gamma.\Gamma\mid  P \ar{d}{\pi_{\Delta.\Gamma.\Gamma}}  \\
U\Delta.U\Gamma \ar{r}{\delta_{\Delta,\Gamma}} & U\Delta.U\Gamma.U\Gamma
\end{tikzcd}
$$

\end{itemize}
Moreover, one has $\delta_{\Delta.\Gamma}^{\sharp}\circ \pi_{\Delta.\Gamma}^{\sharp}= \mathrm{id}_{\BB C_{\D P}^{\Delta.\Gamma}}$.

In $\CTX$ the operation $\pi_{\Delta.\Gamma}^{\sharp}$ turns a predicate $\C C(z,x)$ into a predicate $\C C^{\dag}(z,x,y)=\C C(z,x)$ by adding a ``dummy'' variable $y$; the operation $\delta_{\Delta.\Gamma}^{\sharp}$ turns a predicate $\C C(z,x,y)$ into a predicate $\C C^{\ddag}(z,x)=P(z,x,x)$. Notice that $(\C C^{\dag})^{\ddag}=\C C$.
%
%
%

The difference structure is provided by the following data:
\begin{itemize}
\item for all objects $\Gamma,\Delta$ of $\BB C_{0}$ and for all predicate $Q\in \D P(\Delta.\Gamma.\Gamma)$, an arrow $r_{\Delta,\Gamma\mid Q}: 
\Delta.\Gamma\mid \delta_{\Delta.\Gamma}^{\sharp}(Q) \to \Delta.\Gamma.\Gamma \mid Q\times \Dist\Gamma$, 
such that the composition of $r_{\Delta,\Gamma\mid Q}$ with the projection $\pi_{1}:
\Delta.\Gamma.\Gamma\mid Q\times \Dist \Gamma\to \Delta.\Gamma.\Gamma\mid Q$ coincides with 
$\delta_{\Delta,\Gamma}^{+}:\Delta.\Gamma\mid \delta_{\Delta.\Gamma}^{\sharp}(Q) \to \Delta.\Gamma.\Gamma\mid Q$

\item for all objects $\Delta,\Gamma$ of $\BB C_{0}$, predicates 
$Q\in \D P(\Delta.\Gamma.\Gamma)$ and  $P=f^{\sharp}P^{\flat}\in \D P(\Delta.\Gamma.\Gamma)$, and for any arrow  $c: \Delta.\Gamma\mid \delta_{\Delta.\Gamma}^{\sharp}(Q) \to \Delta.\Gamma.\Gamma \mid Q\times P$ making the diagram below commute
$$
\begin{tikzcd}
\Delta.\Gamma\mid \delta_{\Delta.\Gamma}^{\sharp}(Q)\ar{d}[left]{r_{\Delta,\Gamma\mid Q}} \ar{rr}{c} & & \Delta.\Gamma.\Gamma\mid Q\times P \ar{d}{\pi_{1}} \\
\Delta.\Gamma.\Gamma\mid Q\times \Dist\Gamma 
\ar{rr}{\pi_{1}} & & 
\Delta.\Gamma.\Gamma\mid Q
\end{tikzcd}
$$
(notice that this implies that $c$ is of the form $\langle \delta_{\Delta,\Gamma}^{+}(Q),c'\rangle$) a choice of a diagonal filler $j_{\Delta,\Gamma,Q,f,P,c}$ making both triangles commute. 
\end{itemize}
We require the data above to satisfy a few coherence conditions, namely that for all $g\in \BB C_{0}(\Sigma,\Delta)$ and vertical morphism 
$h\in \BB C_{\D P}^{\Delta.\Gamma.\Gamma} (R , Q)$, 
 $\Dist \Gamma=(g.\Gamma.\Gamma)^{\sharp}\Dist \Gamma$, and the pullback diagrams below commute:
$$
\begin{tikzcd}
\Sigma.\Gamma\mid \delta_{\Delta.\Gamma}^{\sharp}(g^{\sharp}Q) \ar{d}[left]{r_{\Sigma,\Gamma\mid Q}} \ar{rr}{\delta_{\Delta.\Gamma}^{\sharp}(g^{+})}&	&\Delta.\Gamma\mid \delta_{\Delta.\Gamma}^{\sharp}(Q) \ar{d}{r_{\Delta,\Gamma\mid Q}} \\
\Sigma.\Gamma.\Gamma\mid g^{\sharp}Q\times \Dist \Gamma \ar{rr}{(Ug.\Gamma.\Gamma)^{+}}	& &		\Delta.\Gamma.\Gamma\mid Q \times \Dist \Gamma 
\end{tikzcd}
$$
$$
\begin{tikzcd}
\Delta.\Gamma\mid \delta_{\Delta.\Gamma}^{\sharp}(R) \ar{d}[left]{r_{\Delta,\Gamma\mid R}} \ar{rr}{ \delta_{\Delta.\Gamma}^{\sharp}(h)}&	&\Delta.\Gamma\mid \delta_{\Delta.\Gamma}^{\sharp}(Q) \ar{d}{r_{\Delta,\Gamma\mid Q}} \\
\Delta.\Gamma.\Gamma\mid R\times \Dist \Gamma \ar{rr}{ h\times \Dist\Gamma}	& &		\Delta.\Gamma.\Gamma\mid Q \times \Dist \Gamma 
\end{tikzcd}
$$
$$
\begin{tikzcd}
\Sigma.\Gamma.\Gamma \mid g^{\sharp}Q\times \Dist \Gamma \ar{rrr}{j_{\Sigma,\Gamma,g^{\sharp}Q,gf, P, c^{*}}} \ar{d}[left]{g.\Gamma.\Gamma^{+}}
& & & \Sigma.\Gamma.\Gamma\mid g^{\sharp}Q\times g^{\sharp}P\ar{d}{ g.\Gamma.\Gamma^{+} } \\
\Delta.\Gamma.\Gamma \mid Q\times \Dist \Gamma \ar{rrr}{j_{\Delta,\Gamma,Q,f,P,c}} & & & \Delta.\Gamma.\Gamma\mid Q\times P 
\end{tikzcd}
$$
$$
\begin{tikzcd}
\Delta.\Gamma.\Gamma \mid R\times \Dist \Gamma \ar{rrr}{j_{\Delta,\Gamma,R,f,\langle \delta_{\Delta,\Gamma}^{+},c'\circ h\rangle}} \ar{d}[left]{h\times \Dist \Gamma}
& & & \Delta.\Gamma.\Gamma\mid R\times P\ar{d}{ h\times P} \\
\Delta.\Gamma.\Gamma \mid Q\times \Dist \Gamma \ar{rrr}{j_{\Delta,\Gamma,Q,f,P, c}} & & & \Delta.\Gamma.\Gamma\mid Q\times P 
\end{tikzcd}
$$
where the arrow $c^{*}: \Sigma.\Gamma\mid \delta_{\Sigma.\Gamma}^{\sharp} (g^{\sharp}Q)\to \Sigma.\Gamma.\Gamma\mid g^{\sharp}Q\times g^{\sharp}P$ is given by the universality of the pullback along $g$:
$$
\begin{tikzcd}
\Sigma.\Gamma\mid \delta_{\Sigma.\Gamma}^{\sharp}(g^{\sharp}Q) \ar{rr}{(g.\Gamma.\Gamma)^{+}} 
\ar{dd}{\delta_{\Sigma.\Gamma}^{+}}\ar[dashed]{rd}{c^{*}}
& & \Delta.\Gamma\mid \delta_{\Delta.\Gamma}^{\sharp}(Q) \ar{d}{c}\\
 & 	\Sigma.\Gamma.\Gamma\mid g^{\sharp}Q\times g^{\sharp} P \ar{d}[left]{\pi_{\Sigma.\Gamma.\Gamma}}			\ar{r}{(g.\Gamma.\Gamma)^{+}} & \Delta.\Gamma.\Gamma\mid Q\times P \ar{d}{\pi_{\Delta.\Gamma.					\Gamma}} \\
\Sigma.\Gamma.\Gamma\mid g^{\sharp}Q\ar{r}{\pi_{\Sigma.\Gamma.\Gamma}} &	\Sigma.\Gamma.\Gamma   \ar{r}{g.\Gamma.\Gamma} & \Delta.\Gamma.\Gamma
\end{tikzcd}
$$
These (admittedly complicated) conditions essentially say that $\Dist \Gamma$, $r_{\Delta,\Gamma}$ and $j_{\Delta,\Gamma,Q,f,P,c}$ ``do not depend on'' $\Delta$ and $Q$, i.e.~are invariant under substitutions of the variables in $\Delta$ and $Q$. In other words, they assure the soundness of the equations below for the substitution operation:
\begin{align*}
\partial(t)[v/y] & = \partial(t[v/y]) \\
\partial(t)[c/\epsilon] & = \partial(t) \\
\J (t,u,a,[x]b)[v/y] & = \J (t[v/y],u[v/y],a[v/y], [x]b[v/y])  \\
\J (t,u,a,[x]b)[c/\epsilon] & = \J (t,u,a[c/\epsilon], [x]b[c/\epsilon]) 
\end{align*}

\begin{proof}[Proof of Proposition \ref{prop:interpretation}]
The definition of the functor $\model{\_ }_{m}:\CTX_{0}\to \BB C_{0}$ is standard. 
The functor $\nudel{\_ }_{m}: \CTX\to \BB C$ is defined as follows:

\begin{itemize}
\item to any pure predicate $(x,y\in A)\C C(x,y)$ we associate $\nudel{\C C}_{m}\in \D P^{\flat}(\model{A}_{m})$ by induction as follows:

	\begin{itemize}
	\item if $\C C(x,y)=D_{A}(x,y)$, then $\nudel{A}_{m}=\Dist{ (\model{A}_{m})}$;
	\item if $\C C(x,y)=\C C_{1}(\pi_{1}(x),\pi_{1}(y))\times \C C_{2}(\pi_{2}(x),\pi_{2}(y))$, then $\nudel{\C C}_{m}= \pi_{1}^{\sharp}\nudel{\C C_{1}}_{m}\times \pi_{2}^{\sharp}\nudel{\C C_{2}}_{m}$;
	\item if $\C C(x,y)=(\Pi z\in B)\C C'(xz,yz)$, then $\nudel{\C C}_{m}= \Pi_{\model{B}_{m}}\nudel{\C C'}_{m}$;
	\item if $\C C(x,y)=(\Pi w,w'\in B)(D_{B}(w,w')\to \C C'(xww',yww'))$, then 
	$\nudel{\C C}_{m}=\nudel{\C C'}_{m}^{\left( \Dist{\model{B}_{m}} \right)} $.
%
	
%

	\end{itemize}
	
\item For any arrow $(\B t\mid a): (\B x\in \Phi_{0}\mid \BS\epsilon \in \Phi_{1}(\B x))\to (\B y\in \Psi_{0}\mid \delta\in \C C(\B y))$ we define an arrow $(\model{t}_{m}\mid \nudel{a}_{m}): (\model{\Phi_{0}}_{m}\mid \nudel{\Phi_{1}}_{m})\to (\model{\Psi_{0}}_{m}\mid (\model{t}_{m})^{\sharp} \nudel{\C C}_{m})$ in $\BB C$, 
where $\nudel{a}_{m}$ is a vertical morphism in $\BB C_{\D P}^{\model{\Phi_{0}}_{m}}$, by induction on $a$. We here only mention the cases related to the difference type:
	\begin{itemize}
	\item if $a=\partial(t)$, where 
	$(\B xt\in A$, $\nudel{a}_{m}= r_{\model{\Phi'_{0}}_{m}, \model{A}_{m}\mid \nudel{\Phi_{1}}_{m}}\circ \model{t}_{m}$, where 
	$\B x\in \Phi_{0}= \B x'\in \Phi'_{0}, x\in A$;
	
	\item if $a=\J (t,u,b,[x]c)$, 
where 
	 $a\in D_{A}(t,u)$,  $(\B x, \B y\in \Phi_{0})\C C(\B x, \B y)=\C C^{\flat}(\B t(\B x, \B y))$ and 
	 $(x\in A)c\in \C C(x,x )$,
	then  $\nudel{a}= j_{\model{\Phi_{0}}_{m},\model{A}_{m},\model{\Phi_{1}}_{m}, \model{\B t}_{m},\nudel{\C C^{\flat}}_{m},\nudel{c}_{m}}$.	\end{itemize}
\end{itemize}
\end{proof}

\begin{proof}[Proof of Theorem \ref{thm:wfs}]
%

We must show that $i_{f}\in \C L_{\D P}$, so let $P\in \D P(\Sigma)$ be a predicate. The problem of finding a diagonal filler for a commutative diagram of the form
$$
\begin{tikzcd}
\Gamma \ar{d}{i_{f}}\ar{rr}{d} & & \Sigma\mid P \ar{d}{\pi_{\Sigma}} \\
\Gamma.\Delta\mid \Delta_{f}\ar{rr}{h} & & \Sigma
\end{tikzcd}
$$
can be reduced, by pulling back along $h$, to that of finding a diagonal filler $j$ for a diagram of the form
$$
\begin{tikzcd}
\Gamma \ar{d}{i_{f}}\ar{rr}{c} & & \Gamma.\Delta\mid Q \ar{d}{\pi_{\Gamma.\Delta}} \\
\Gamma.\Delta\mid \Delta_{f}\ar{rr}{\pi_{\Gamma.\Delta}} & & \Gamma.\Delta
\end{tikzcd}
$$
(with $Q=h^{\sharp}P$). We will obtain $j$ by putting together a few commutative diagrams:
\begin{enumerate}
\item 
$$
\begin{tikzcd}
\Gamma \ar{rr}{i_{f}.c}\ar{d}{i_{f}} & &   \Gamma.\Delta\mid \Delta_{f}. Q \ar{d}{
r_{\Gamma\mid \Delta_{f}.Q, \Delta}  } \\
\Gamma.\Delta\mid\Delta_{f} \ar{rr}[below]{
\delta_{\Delta}^{+}. (\delta_{\Delta}^{+}\circ c). r_{\Gamma.\Delta}}  & & \Gamma.\Delta.\Delta\mid\pi_{\Gamma.\Delta}^{\sharp}(\Delta_{f}). \pi_{\Gamma.\Delta}^{\sharp}(Q).\Dist\Delta
\end{tikzcd}
$$

\item
$$
\begin{tikzcd}
\Gamma \ar{d}{i_{f}.c}  \ar{rr}{c} && \Gamma.\Delta\mid Q\\
\Gamma.\Delta\mid \Delta_{f}.Q \ar{rr}[below]{ 
\delta_{\Delta}^{+}. (\delta_{Q}\circ \delta_{\Delta}^{+})
  } & &  \Gamma.\Delta.\Delta\mid \pi_{\Gamma.\Delta}^{\sharp}(\Delta_{f}).\pi^{\sharp}_{\Gamma.\Delta}(Q).\pi^{\sharp}_{\Gamma.\Delta}(Q) \ar{u}{\pi_{\Gamma.\Delta}^{+}\circ \pi_{3}}
\end{tikzcd}
$$
where $\pi_{\Gamma.\Delta}$ is the projection $\Gamma.\Delta.\Delta\to \Gamma.\Delta$ and we use the fact
that $\delta_{\Delta}^{\sharp}(\pi_{\Gamma.\Delta}^{\sharp}(\Delta_{f}))=\Delta_{f}$ and $\delta_{\Delta}^{\sharp}(\pi_{\Gamma.\Delta}^{\sharp}(Q))=Q$, so that 
$\delta_{\Delta}^{+}: \Gamma.\Delta\mid \Delta_{f}\to \Gamma.\Delta.\Delta\mid \pi^{\sharp}_{\Gamma.\Delta}(\Delta_{f})$ and 
$(\delta_{\Delta})^{+}: \Gamma.\Delta\mid Q\to \Gamma.\Delta.\Delta\mid \pi_{\Gamma.\Delta}^{\sharp}(Q)$.

\item 
$$
\begin{tikzcd}
\Gamma.\Delta.\Delta\mid \pi_{\Gamma.\Delta}^{\sharp}(\Delta_{f}). \pi_{\Gamma.\Delta}^{\sharp}(Q).\pi_{\Gamma.\Delta}^{\sharp}(Q) 
\ar{d}{\pi_{\Gamma.\Delta.\Delta}}
\ar{rrr}{
(\pi_{\Gamma.\Delta}^{+}\circ \pi_{2}).
(\pi_{\Gamma.\Delta}^{+}\circ \pi_{3})
} & & &  \Gamma.\Delta\mid \Delta_{f}.Q \ar{d}{\pi_{\Gamma.\Delta}} \\
\Gamma.\Delta.\Delta \ar{rrr}[below]{\pi_{\Gamma.\Delta}}  & &  & \Gamma.\Delta
\end{tikzcd}
$$

By putting all this together we can obtain a diagonal filler from the diagram illustrated in Fig.~\ref{fig:landscape}, where $e=\delta_{\Delta}^{+}. (\delta_{Q}\circ \delta_{\Delta}^{+}$, and where the central diagonal filler exists by hypothesis.

\end{enumerate}

 If $h\in \D P$, then by definition of $\C L_{\D P}$ it has the right-lifting property with respect to all $f\in \C L_{\D P}$; we deduce then $\D P^{*}\subseteq \C R_{\D P}$. This proves that $p_{f}\in \C R_{\D P}$.

Since $\C L_{\D P}^{\tri}=\C R_{\D P}$ holds by definition, it remains to prove that $\C L_{\D P}=^{\tri}\C R_{\D P}$. On the one hand, from $\D P\subseteq \C R_{\D P}$, we deduce $^{\tri}\C R_{\D P}\subseteq ^{\tri}\D P\subseteq \C L_{\D P}$. For the converse direction, 
by the Retract Argument (Lemma 1.1.9 in \cite{Hovey1999}), any $g\in \C R_{\D P}$ is a retract of a projection: 
from $g=p_{g}\circ i_{g}$, and the fact that $g$ has the right lifting property with respect to $i_{g}$, we deduce that there exists $h$ (a diagonal filler of $ p_{g}\circ i_{g}= g\circ \mathrm{id}_{X}$) such that $g\circ h= p_{g}$. 
This implies in particular that there is a diagram of the form
$$
\begin{tikzcd}
X\ar[bend left]{rr}{\mathrm{id}_{X}} \ar{d}{g} \ar{r}{i_{g}} & A \ar{d}{p_{g}} \ar{r}{h} & X\ar{d}{g} \\
Y \ar[-,double]{r}{} & Y \ar[-, double]{r} & Y
\end{tikzcd}
$$
Using this we can show that given a diagram of the form
$$
\begin{tikzcd}
 C \ar{d}{f} \ar{r}{k} & X \ar{d}{g}      \\
 D \ar{r}{k'} & Y
\end{tikzcd}
$$
where $f\in \C L_{\D P}$, we can construct a diagram 
$$
\begin{tikzcd}
 C \ar{d}{f} \ar{r}{k} & X  \ar[bend left]{rr}{1_{X}}\ar{d}{g}    \ar{r}{i_{g}} & A \ar{r}{h} \ar{d}{p_{g}} & X \ar{d}{g}\\
 D
 \ar[dashed]{rru}{j}
  \ar{r}{k'} & Y \ar[-, double]{r} & Y \ar[-, double]{r} & Y
\end{tikzcd}
$$
where the diagonal filler $j$ exists since $f$ has the left-lifting property with respect to $p_{g}$, and from which we obtain a diagonal filler $j'= h\circ j$ for the original diagram. 
We have thus shown that any $f\in \C L_{\D P}$ has the left-lifting property with respect to any $g\in\C R_{\D P}$, and thus 
$\C L_{\D P}\subseteq ^{\tri}\C R_{\D P}$. 
\end{proof}

\begin{landscape}
\begin{figure}
\adjustbox{center, scale=1}{
$
\begin{tikzcd}
& & & (2)  & & & & \\
\Gamma 
\ar[bend left=20]{rrrrrrr}{c}
\ar{dd}{i_{f}}  \ar{rr}{i_{f}.c} && \Gamma.\Delta\mid \Delta_{f}.Q \ar{dd}[left]{i_{\Gamma\mid\Delta_{f}.Q, \Delta}}    \ar{rr}{e}  & &  \Gamma.\Delta.\Delta\mid \pi_{\Gamma.\Delta}^{\sharp}(\Delta_{f}).\pi_{\Gamma.\Delta}^{\sharp}(Q).\pi_{\Gamma.\Delta}^{\sharp}(Q)
\ar{dd}{\pi_{\Gamma.\Delta.\Delta}}
 \ar{rrr}{
(\pi_{\Gamma.\Delta}^{+}\circ \pi_{2}).
(\pi_{\Gamma.\Delta}^{+}\circ \pi_{3})
} & & & \Gamma.\Delta\mid \Delta_{f}.Q \ar{dd}{\pi_{\Gamma.\Delta}} \\
& (1) & & & &(3) & & \\
\Gamma.\Delta\mid \Delta_{f} \ar[bend right=20]{rrrrrrr}{\pi_{\Gamma.\Delta}} \ar{rr}[above]{ 
\delta_{\Delta}^{+}. (\delta_{Q}\circ c).i_{\Gamma.\Delta}
  } & &  \Gamma.\Delta.\Delta\mid \pi_{\Gamma.\Delta}^{\sharp}(\Delta_{f}).\pi^{\sharp}_{\Gamma.\Delta}(Q).\Dist\Delta
  \ar[dashed]{rruu}{j_{\Gamma\mid \pi^{\sharp}_{\Gamma.\Delta}(\Delta_{f}.Q), \Delta, Q,e}}
    \ar{rr}[below]{\pi_{\Gamma.\Delta.\Delta}}
     &  &   \Gamma.\Delta.\Delta \ar{rrr}{\pi_{\Gamma.\Delta}}
   & & & \Gamma.\Delta
\end{tikzcd}
$}
\caption{Construction of the diagonal filler.}
\label{fig:landscape}
\end{figure}
\end{landscape}

\section{Metric Preservation: Details}

We will need the following generalization of the notion of metric space:
a \emph{parameterized (pseudo-)metric space} (PMS) is a triple $(X,K,a)$, where $a:X\times X\to (\BB R_{\geq 0})^{K}$ satisfies, for all $k\in K$:
\begin{align*}
a(x,x)(k)& = 0 \\
a(x,y)(k)&= a(y,x)(k)\\
a(x,y)(k) & \leq a(x,z)(k)+a(z,y)(k)
\end{align*}
Usual (pseudo-)metric spaces can be identified with PMS $(X,K,a)$ where $K$ is a singleton.

We recall that $\Met $ is the category of pseudo-metric spaces and non-expansive map. $\Met$ is symmetric monoidal closed; its monoidal product is $(X,a)\otimes (Y,b)=(X\times Y, a+b)$ and the right-adjoint to $\otimes $ is $[(X,a),(Y,b)]=(\Met(X,Y), b_{\sup})$, where $b_{\sup}(f,g)=\sup\{b(f(x),g(x))\mid x\in X\}$.
It is a standard fact that $\Met$ has enough structure to interpret ST$\lambda$C$^{!}$, with $!_{r}A$ corresponding to the rescaling of a metric space by $r$.

%
%
%
%

We now describe the $\DTT$-structure associated with the forgetful functor $U:\Met \to \Set$.

For any set $X$, the class of pure predicates $\D P^{\flat}(X)$ is made of all PMS of the form $P=(X,K,a)$, with $(X\otimes X\mid P)=\coprod_{x,x'\in X}\prod_{k\in K}\widetilde a(x,x')(k)$ and projection  $\pi_{X}:(X\otimes X\mid P)\to X\times X$. 
The class $\D P(X)$ contains all pullbacks  $\langle f,g\rangle^{\sharp}P$, for all $P=(X,K,a)\in \D P^{\flat}(X)$ and  
$\langle f,g\rangle\in \Met(Y, X\otimes X)$, with $(Y\mid\langle f,g\rangle^{\sharp}P)=\coprod_{ y\in Y}\prod_{k\in K}\widetilde a(f(y),g(y))(k)$ and 
 projection $\pi_{Y}: (Y\mid \langle f,g\rangle^{\sharp}P)\to Y$.

Observe that for any projection $\pi: X\times Y\to X$ and PMS $(X,K, a)$, $(X, K,\pi^{\sharp}a)$ is still a PMS, where $\pi^{\sharp}a(u,v)(k)=a(\pi_{1}(u), \pi_{1}(v))(k)$ (in fact, $\pi^{\sharp}a(u,u)=a(\pi_{1}(u), \pi_{1}(u))=0$, $\pi^{\sharp}a(u,v)=a(\pi_{1}(u), \pi_{1}(v))=a(\pi_{1}(v), \pi_{1}(u))=\pi^{\sharp}a(v,u)$, and 
$\pi^{\sharp}a(u,v)=a(\pi_{1}(u), \pi_{1}(v))\leq a(\pi_{1}(u), \pi_{1}(w))+a(\pi_{1} (w), \pi_{1}(v))= \pi^{\sharp}a(u,w)+\pi^{\sharp}a(w,v)$).

Moreover, for all PMS $(X\times X, K,a)$, also $(X, K, \delta_{X}^{\sharp}a)$ is a metric space, where 
$\delta_{X}^{\sharp} a(x,y)(k)=a(\langle x,x\rangle, \langle y,y\rangle)(k)$ (
in fact 
$\delta_{X}^{\sharp} a(x,x)(k)=a(\langle x,x\rangle, \langle x,x\rangle)(k)=0$, 
$\delta_{X}^{\sharp} a(x,y)(k)=a(\langle x,x\rangle, \langle y,y\rangle)(k)= a(\langle y,y\rangle, \langle x,x\rangle)(k)=
\delta_{X}^{\sharp} a(y,x)(k)$, and 
$\delta_{X}^{\sharp} a(x,y)(k)=a(\langle x,x\rangle, \langle y,y\rangle)(k)\leq 
a(\langle x,x\rangle, \langle z,z\rangle)(k)+a(\langle z,z\rangle, \langle y,y\rangle)(k)=\delta_{X}^{\sharp} a(x,z)(k)+\delta_{X}^{\sharp} a(z,y)(k)$).

An arrow in the slice category $\Set_{\D P}^{X}$ between $\coprod_{x\in X}\prod_{k\in K}\widetilde a(f(x),g(x))(k)$ and \\ 
$\coprod_{x\in X}\prod_{h\in H}\widetilde b(f'(x),g'(x))(h)
$ 
is given by a function $\varphi:  X \times (\BB R_{\geq 0})^{K}\times H \to \BB R_{\geq 0}$ such that 
$\varphi(x, \phi,h)\in \widetilde b(f'(x), g'(x))(h)$. 
Given PMS $P=(Y,K,a)$ and $Q=(Z,H,b)$, let $P*Q=(Y\times Z, K+H, a+b)$, where 
$a+b(\langle y,z\rangle, \langle y',z'\rangle)(\langle 0,k\rangle)= a(y,y')(k)$ and 
$a+b(\langle y,z\rangle, \langle y',z'\rangle)(\langle 1,h\rangle)=b(z,z')(k)$.
One can check that:
\begin{itemize}

\item for all predicates $P=\langle f,g\rangle^{\sharp}P'$ and $Q= \langle f',g'\rangle^{\sharp}Q'\in \D P(X)$, 
where $P'\in \D P^{\flat}(Y)$ and $Q'\in \D P^{\flat}(Z)$ are
  the PMS $(Y,K,a)$ and $(Z,H,b)$, 
 the product $P\times Q\in \D P(X)$ is $\langle \langle f, f'\rangle,\langle g,g'\rangle\rangle^{\sharp}(P*Q)$, with  
 $(X\mid P\times Q)=\coprod_{x\in X} \prod_{u\in K+H}(\widetilde a(f(x),g(x))+\widetilde b(f'(x),g'(x)))(u)$.
Observe that if $P,Q\in \D P^{\flat}(X)$, then $P\times Q=P*Q\in \D P^{\flat}(X)$.

%
\item for any pure predicate $P=(X,K,a)\in \D P^{\flat}(X)$ and set $I$, the dependent product $\Pi_{I}P\in \D P^{\flat}(X^{I})$ is the PMS $(X^{I},  I\times K, \Pi_{I}a)$ where $(\Pi_{I}a)(f,g)(\langle i,k\rangle)=a(f(i), g(i))(k)$;

\item for all PMS $P=(X,K,a)$ and $Q=(Y,K,b)\in \D P(Y)$, the  dependent product $Q^{\left (\Dist X\right)}\in \D P^{\flat}(Y^{X\times X})$ is the PMS $(Y^{X\times X}, H\times (X\otimes X \mid P), \F p_{a,b})$, where $\F p_{a,b}(f,g)(\langle h, \langle\langle x,y\rangle, \phi\rangle\rangle)=b(f(x,y), g(x,y))(h)$.

%

\end{itemize}
%
 
%

The difference structure is as follows:
\begin{itemize}
\item for all metric spaces $(X,a)$, $\Dist X\in \D P^{\flat}(X)$ is $(X,a)$ itself, seen as a PMS;

\item for all metric spaces $(X,a)$, and set $Y$, predicate $P=\langle f,g\rangle^{\sharp}P'\in \D P(Y\otimes (X \otimes X))$, (with $(Y\otimes X\otimes X\mid P)=\coprod_{y\in Y, x,x'\in X}\prod_{k\in K}\widetilde a(f(y,x,x'),g(y,x,x'))$, 
the pullback $(Y\otimes X\mid \delta_{X}^{\sharp}P)$ is 
$\coprod_{y\in Y, x\in X}\prod_{k\in K}\widetilde a(f(y,x,x),g(y,x,x))$; 
the morphism $r_{Y,X \mid P}:(Y\otimes X\mid \delta^{\sharp}_{X}P) \to (Y\otimes X\otimes X\mid P\times \Dist X)$ is 
given by $r_{Y,X\mid P}(\langle \langle y,x\rangle, \phi\rangle)(k)=\langle \langle y,x,x\rangle, \phi+0\rangle$.

%

\item for all metric spaces $(X,a)$, $(Y,b)$ predicate $P=\langle m,n\rangle^{\sharp}P'\in \D P(Y\otimes X\otimes X)$ (with $(Y\otimes X\otimes X\mid P)=\coprod_{y\in Y, x,x'\in X}\prod_{h\in H}\widetilde b(m(y,x,x'), n(y,x,x'))(h)$, pure predicate
$Q=(Z,K,e)\in \D P^{\flat}(Z)$, non-expansive functions $\langle f,g\rangle\in \Met (Y\otimes X\otimes X, Z\otimes Z)$ and for any 
 function $c: (Y'\otimes X\mid \delta_{X}^{\sharp}P )\to( Y\otimes X\otimes X\mid P \times \langle f,g\rangle^{\sharp}Q)$ satisfying 
 $c(\langle\langle y,x\rangle, \phi\rangle)  =  \langle\langle y, x,x\rangle,c'(y,x,\phi) \rangle$, with $c'(y,x,\phi): H+K \to \BB R_{\geq 0}$ satisfying $c'(y,x,\phi)(\langle 0,k\rangle)=\phi(k)$ and $c'(y,x,\phi)(\langle 1,h\rangle)\in \widetilde e(f(y,x,x),g(y,x,x))$, for all $y\in Y'$, $x,x'\in X$, $\phi\in ( \BB R_{\geq 0})^{K}$ with $\phi(k)\geq  b'(x,x')(k)$ and 
$r\geq a(x,x')$, we let
$\big(j_{Y,X,Q,\langle f,g\rangle,P,c}(\langle\langle y,x,x'\rangle, \langle \phi ,r\rangle\rangle)\big )_{2}= \psi_{y,x,x',\phi,r}: H+K\to \BB R_{\geq 0}$, where 
$\psi_{y,x,x',\phi,r}(\langle 0,k\rangle)=\phi(k)$, and $\psi_{y,x,x',\phi,r}(\langle 1,h\rangle)= c'(y,x,\phi)(\langle 1,h\rangle)+ r$.

In fact, since $\phi(k)\geq b'(x,x')(k)\geq b'(x,x)(k)$, we have  $c'(y,x,\phi)(\langle 1,h\rangle)\geq e(f(y,x,x),g(y,x,x))(h)$; moreover, 
from $\langle f,g\rangle\in \Met (Y\otimes X\otimes X, Z\otimes Z)$ it follows 
$e(f(y,x,x'),f(y,x,x))+ e(g(y,x,x),g(y,x,x'))\leq b(y,y)+a(x,x)+a(x,x') =a(x,x')\leq r$; then, using the triangular law, we deduce $c'(y,x,\phi)(\langle 1,h\rangle) + r\geq  e'(f(y,x,x'),g(y,x,x')(h)$.
\end{itemize}

Observe that $((j_{Y,X,Q,\langle f,g\rangle,P,c}\circ r_{Y,X\mid P})(x,y,\phi))_{2}=c'(y,x,\phi)$.
Moreover, when $Q=\Dist X$, $f(y,x,x')=x$, $g(y,x,x')=x'$ and $c'(y,x,\phi)(\langle 1,h\rangle)=0$, 
we deduce that $j_{Y,X,Q,\langle f,g\rangle,P,c}=\mathrm{id}_{Y\times X\times X\mid Q\times \Dist X}$.
%
 Hence both the $\beta$- and $\eta$-rules are valid. 

We leave to the reader to check that the validity of all required coherence conditions.

\subparagraph*{Extensionality}
This model does not satisfy any of the extensionality axioms discussed in Section \ref{sec2}. In fact, the 
type $D_{A\otimes B}$ is interpreted by the PMS corresponding to the 
monoidal product in $\Met$ of (the interpretation of) $D_{A}$ and $D_{B}$, which does \emph{not} coincide with the product in $\D P(\model A\times \model B)$. 

Similarly, the type $D_{A\multimap B}$ is interpreted by the PMS corresponding to the right-adjoint to the monoidal product in $\Met$, which does \emph{not} coincide with either of the two dependent products of predicates.

\section{Differential Logical Relations: Details}

A DLR is a triple $(X,L,\rho)$ where $X$ is a set, $L$ is a complete lattice and $\rho \subseteq X\times L \times X$. A map of DLR $(X,L,\rho )$ and $(Y,M,\mu)$ is a pair $(f,\varphi)$, where $f:X\to Y$ and $\varphi:X\times X\times L\to  M$ is such that $\rho(x,\epsilon,y)$ implies
$\mu(f(x),\varphi(x,y,\epsilon), f(y))$ and $\mu(f(y), \varphi(x,y,\epsilon), f(x))$.
This notion of map is a slight variation with respect to \cite{dallago}, where the auxiliary map $\varphi$ goes from $X\times L$ to $ M$. Yet, this change does not affect the higher-order structure of DLR (see below), and the two families of maps are related by a retraction $\begin{tikzcd}M^{X\times L} \ar[bend left=5]{r}{h} & \ar[bend left=5]{l}{k}M^{X\times X\times L}\end{tikzcd}$ where $h(\varphi)(x,y,\epsilon)=\varphi(x,\epsilon)$ and $k(\psi)(x,\epsilon)=\sup\{\psi(x,y,\epsilon)\mid y\in X\land  \rho(x,\epsilon,y)\}$. 

DLR and their maps form a category $\DLR$, where the identity of a DLR $(X,L,\rho)$ is the map 
$(f, \lambda xy\epsilon.\epsilon)$, and composition of $(f,\varphi)$ and $(g,\psi)$ is $(g\circ f, x,y,\epsilon\mapsto\psi(f(x),f(y), \psi(x,y,\epsilon))$. 

$\DLR$ is cartesian closed, with the product and exponential of DLR $(X,L,\rho)$ and $(Y,M,\mu)$ being the DLR
$(X\times Y, L\times M, \rho\times \mu)$ and $(Y^{X}, M^{X\times X\times L}, \rho^{(X\times X\times \mu)})$, where $\rho^{(X\times X\times \mu)}(f,\varphi,g)$ holds if 
$\rho(x,\epsilon,y) $ implies $\mu(f(x), \varphi(x,y,\epsilon), f(y))$,
$\mu(f(x), \varphi(x,y,\epsilon), g(y))$,
$\mu(g(x), \varphi(x,y,\epsilon), f(y))$ and $\mu(g(x), \varphi(x,y,\epsilon), g(y))$.

We now describe the $\DTT$-structure associated with the forgetful functor $U:\DLR \to \Set$.

For any set $X$, $\D P^{\flat} (X)$ is made of all DLR of the form $P=(X,L,\rho)$, with  $X\times X\mid P=\coprod_{ x,x'\in X}\widetilde \rho(x,x')$, and projection $\pi_{X}:(X\times X\mid P)\to X\times X$.
$\D P(Y)$ is made of all pullbacks $\langle f,g\rangle^{\sharp}P$, where $P=(X,L,\rho)\in \D P^{\flat}(X)$, and  $f,g:Y\to X$, with $(Y\mid \langle f,g\rangle^{\sharp}P)=\coprod_{ y\in Y}\widetilde \rho(f(y), g(y))$ and projection $\pi_{Y} :(Y\mid \langle f,g\rangle^{\sharp}P)\to Y$.
 
%

An arrow in the slice category $\Set_{\D P}^{X}$ between $\coprod_{x\in X}\widetilde \rho(f(x),g(x))$ and 
$\coprod_{x\in X}\widetilde \mu(f'(x),g'(x))$ (for given DLR $(X,L,\rho)$ and $(Y,M,\mu)$) is given by a function 
$\varphi:  X \times L \to M$  such that $\rho(f(x),\epsilon,g(x))$ implies $\mu(f'(x), \varphi(x,\epsilon), g'(x))$. One can check that:
\begin{itemize}

\item for all predicates $P=\langle f,g\rangle^{\sharp}P'\in \D P(X)$ and $Q=\langle f',g'\rangle^{\sharp}Q'\in \D P(X)$, their product is the predicate 
$\langle \langle f,f'\rangle, \langle g,g'\rangle\rangle^{\sharp}(P'\times Q')\in \D P(X)$.
Observe that 
%
%
 if $P$ and $Q$ are in $\D P^{\flat}(X)$, then $P\times Q\in \D P^{\flat}(X)$;
%

\item for all pure predicate $P=(X,L,\rho)\in \D P^{\flat}(X)$ and set $I$, the dependent product $\Pi_{I}P\in \D P^{\flat}(X^{I})$ is  the DLR  $(X^{I},L^{I},\rho^{I})$, where 
$\rho^{I}(f,\varphi,g)$ iff $\forall i\in I$, $\rho(f(i),\varphi(i),g(i))$; 

\item for all DLR $(X,L,\rho)$ and pure predicate $Q=(Y,M,\mu)\in \D P^{\flat}( Y)$, the dependent  product $P^{\left (\Dist X\right)}\in \D P(Y^{X\times X})$ is the exponential of $ (X,L,\rho)$ and $Q$ in $\DLR$.
%

\end{itemize}

The difference structure is as follows:
\begin{itemize}
\item for all DLR $(X,L,\rho)$, $\Dist X\in \D P^{\flat}(X)$ is just $(X,L,\rho)$;
\item for all DLR $(X,L,\rho)$, set $Y$ and predicate $P= \langle f,g\rangle^{\sharp}P'\in \D P(Y\times X\times X)$, where $(Y\times X\times X\mid P)=\coprod_{y\in Y,x,x'\in X}\widetilde \mu(f(y,x,x'),g(y,x,x'))$, the pullback $(Y\times X\mid \delta_{X}^{\sharp}P)$ is $\coprod_{y\in Y, x\in X}\widetilde \mu(f(y,x,x),g(y,x,x))$, while $(Y\times X\times X\mid P\times \Dist X)$ is 
$\coprod_{y\in Y,x,x'\in X}\widetilde \mu(f(y,x,x'),g(y,x,x'))\times \widetilde \rho(x,x')$; we let then 
$r_{Y,X\mid P}(\langle \langle y,x\rangle, s\rangle)= \langle \langle y,x,x\rangle,\langle s, \| x\|_{\rho}\rangle\rangle$;

\item for all DLR $(X,L,\rho)$, sets $Y,Z,W$, predicate $Q=\langle m,n\rangle^{\sharp}Q'\in \D P(Y\times X\times X)$, with $(Y\times X\times X\mid Q)=\coprod_{y\in Y,x,x'\in X}\widetilde \lambda(m(y,x,x'),n(y,x,x'))\in \D P(Y\times X\times X)$, pure predicate $P=(Z,M, \mu)\in \D P^{\flat}(Z)$ and arrows $f,g: Y\times X\times X\to  Z$, for all morphisms $c: (Y\times X\mid \delta_{X}^{\sharp}Q) \to( Y \times X\times X \mid Q\times \langle f,g\rangle^{\sharp}P)$ such that $c(\langle \langle y,x\rangle, s\rangle)= \langle \langle y,x,x\rangle,\langle s, c'(y,x,s)\rangle\rangle$, with $c'(y,x,s)\in \widetilde \mu(f(y,x,x),g(y,x,x))$, we can define a diagonal filler 
$j: (Y\times X\times X \mid Q\times \Dist X) \to ( Y\times X\times X\mid Q\times \langle f,g\rangle^{\sharp}P)$ by 
$$
(j(\langle \langle y,x,x'\rangle, \langle s,r\rangle\rangle))_{2}=
\langle s, \sup\{c'(y,x,s), \inf\widetilde\mu(f(y,x,w),g(y,x,w)  ) \mid 
w\in X \land \rho(x,r,w) \}  \rangle
$$
We have that $((j\circ r_{Y,X\mid Q})(\langle y,x\rangle))_{2}=\langle s, c'(y,x,s)\rangle$, since 
by separatedness $\rho(x,\|x\|_{\rho}, w)$ implies $w=x$, and $\inf\widetilde\mu(f(y,x,x),g(y,x,x))= \| f(y,x,x),g(y,x,x)\|_{\mu}\leq c'(y,x,s)$.

\end{itemize}
We leave to the reader to check the validity of all required coherence conditions.

\subparagraph*{Extensionality}
The DLR model of $\DTT$ satisfies the extensionality axioms \eqref{cext} and \eqref{fext2}.
%
%
In fact, given simple types $A,B$, the interpretation of $D_{A\times B}$ is generated by the cartesian product in $\DLR$ 
of the interpretations of $D_{A}$ and $D_{B}$, which coincides with the product in $\D P(\model A)$. 
Moreover, the interpretation of $D_{A\to B}$ is generated by the exponential of $D_{A}$ and $D_{B}$ in $\DLR$, which coincides with the pullback along $\langle h,h\rangle$ of the  dependent product of the interpretation of $D_{B}$ and $A$, where $h,h: Y^{X}\to Y^{X\times X}$ is given by $h(f)(x,y)=f(x)$, and this in turn coincides with 
the interpretation of the type $(\Pi x,y\in A)(D_{A}(x,y)\to D_{B}(f(x),g(y)))$. 

One can also check that derivatives satisfy the equational rules \eqref{jca}, \eqref{jcb},\eqref{jc}, \eqref{jl2a} and \eqref{jl2}.

\section{Change Structures: Details}

Change structures (as defined in Section \ref{sec5}) and functions form  a category $\CS$ which is cartesian closed. In particular, given 
 CS $(X,\Delta_{X},\oplus,\ominus)$ and $(Y, \Delta_{Y}, \oplus',\ominus')$, their cartesian product is 
$(X\times Y, \Delta_{X}\times \Delta_{Y}, \oplus\times \oplus', \ominus\times\ominus')$, and 
 their exponential  is$(Y^{X},( \Delta_{Y})^{X\times \Delta_{X}}, \oplus^{*}, \ominus^{*})$, where
$(\Delta_{Y})^{X\times \Delta_{X}}(f)$ contains all functions $\varphi:X\times \Delta_{X}\to \Delta_{Y}$, such that for all $x\in X$ and $\de x\in \Delta_{X}x$, $\varphi(x,\de x)\in \Delta_{X}f(x)$, 
$(f\oplus^{*} \varphi)(x)=f(x)\oplus' \varphi(x, \B 0_{x})$ and 
$(f\ominus^{*}g)(x,\de x)= f(x\oplus \de x)\ominus g(x)$.

We now describe the $\DTT$-structure associated with the forgetful functor $U:\CS \to \Set$.

For any set $X$, $\D P^{\flat}(X)$ is made of CS with base set $X$, where for a CS $P=(X,\Delta_{X},\oplus,\ominus)$, $(X\times X\mid P)=\coprod_{ x,x'\in X}\Delta_{X}(x,x')$, with projection $\pi_{X}:X\times X\mid P\to X\times X$. 
$\D P(Y)$ is made of all pullbacks $\langle f,g\rangle^{\sharp}P$, where $P\in \D P^{\flat}(X)$ is a CS and $f,g:Y\to X$, with $(Y\mid  \langle f,g\rangle^{\sharp}P)=\coprod_{ y\in Y}\Delta_{X}(f(y), g(y))$, and with projection $\pi_{Y} : (Y\mid \langle f,g\rangle^{\sharp}P)\to Y$.
 
%


An arrow in the slice category between $\coprod_{y\in Y}\Delta_{X}(f(x),g(x))$ and $\coprod_{y\in Y}\Delta_{X'}(f'(x),g'(x))$ is given by a function $\varphi: Y\times  \Delta_{X} \to  \Delta_{X'}$, such that whenever $\de x\in \Delta_{X}(f(x),g(x)$, $(\varphi(y,\de x))_{2}\in \Delta_{X'}(f'(x),g'(x))$. One can check that:
\begin{itemize}

\item for all predicates $P= \langle f,g\rangle^{\sharp}P\in \D P(X)$ and $Q=\langle f',g'\rangle^{\sharp}Q'\in \D P(X)$, their product is $P\times Q=\langle \langle f,f'\rangle, \langle g,g'\rangle\rangle^{\sharp}(P'\times Q')$  (where $P'\times Q'$ indicates the product of $P'$ and $Q'$ in $\CS$), with 
$(X\mid P\times Q)= \coprod_{x\in X}\Delta_{Y}(f(x),g(x))\times \Delta_{Y'}(f'(x),g'(x))$, and projection 
$\pi_{X}: (X\mid P\times Q)\to X$.

%

%

\item for all pure predicate $P\in \D P^{\flat}(X)$, where $P=(X,\Delta_{X},\oplus, \ominus)$ and set $I$, the dependent product $\Pi_{I}P\in \D P^{\flat}(X^{I})$ is the CS
$(X^{I}, (\Delta_{X})^{I}, \oplus^{I}, \ominus^{I})$ where 
$(\Delta_{X})^{I}f$ is made by all functions $\varphi: I\to \Delta_{X}$ such that for all $i\in I$, $\varphi(i)\in \Delta_{X}f(x)$, $(f\oplus^{I} \varphi )(i)= f(i)\oplus \varphi(i)$ and $(f\ominus^{I}g )(i)=f(i)\ominus g(i)$;

\item for all CS $(X,\Delta_{X},\oplus,\ominus)$ and pure predicate $Q\in \D P^{\flat}( Y)$, where $Q=( Y, \Delta_{Y}, \oplus', \ominus')$, 
the dependent  product $P^{\left (\Dist X\right)}\in \D P(Y^{X\times X})$
is the CS 
$(Y^{X\times X},( \Delta_{X})^{X\times X\times \Delta_{X}}, \oplus^{*},\ominus^{*})$, where 
$ ( \Delta_{Y})^{X\times X\times \Delta_{X}}(f)$ contains all functions $\varphi:X\times X\times \Delta_{X}\to \Delta_{Y}$ such that for all $x,x'\in X$ and $\de x\in \Delta_{X}(x,x')$, $\varphi(x,x',\de x)\in \Delta_{Y}f(x,x')$, and $(f\oplus^{*}\varphi )(x,x')= f(x,x') \oplus \varphi(x,x', x\ominus x')$, 
$(f\ominus^{*}g)(x,x',\de x)= f(x,x')\ominus g(x,x')$.

\end{itemize}

The difference structure is as follows:
\begin{itemize}
\item for all CS $(X,\Delta_{X},\oplus, \ominus)$, $\Dist X\in \D P^{\flat}(X)$ is the CS itself;

\item for all CS $(X,\Delta_{X},\oplus, \ominus)$, set $Y$ and predicate $P=\langle f,g\rangle^{\sharp}P'\in \D P(Y\times X\times X)$, with $P'=(W,\Delta_{W},\oplus',\ominus')$, the pullback $(Y\times X\mid \delta_{X}^{\sharp}P)$ is 
$\coprod_{y\in Y, x\in X}\Delta_{W}(f(y,x,x),g(y,x,x))$, while 
$(Y\times X\times X\mid P\times \Dist X)$ is 
$\coprod_{y\in Y, x,x'\in X}\Delta_{W}(f(y,x,x'),g(y,x,x'))\times \Delta_{X}(x,x')$; we let then $r_{Y,X\mid P}(\langle \langle y,x\rangle, \de x\rangle)= \langle \langle y,x,x\rangle, \langle \de x, \B 0_{x}\rangle \rangle$;

\item for all CS  $(X,\Delta_{X},\oplus, \ominus)$, sets $Y,Z,W$, predicate
$Q=\langle m,n\rangle^{\sharp}Q'\in \D P(X)$, with $Q'=(W,\Delta_{W},\oplus',\ominus')$, pure predicate 
 $P=(Z, \Delta_{Z},\oplus, \ominus)\in \D P^{\flat}(Z)$, functions $f,g: Y\times X\times X\to Z$ and  
$c:( Y\times X\mid \delta_{X}^{\sharp}Q) \to( Y\times X\times X \mid Q\times \langle f,g\rangle^{\sharp}P)$ such that 
$c(\langle\langle y,x\rangle,\de x\rangle)=\langle \langle y,x,x\rangle,\langle \de x, c'(y,x,\de x)\rangle\rangle$, with $c'(y,x,\de x)\in \Delta_{Z}(f(y,x,x), g(y,x,x))$, we can define a diagonal filler
$j: (Y\times X\times X \mid Q\times \Dist X)\to( Y\times X\times X\mid Q\times \langle f,g\rangle^{\sharp}P)$ by 
\begin{align*}
(j(\langle \langle y,x,x'\rangle, \langle \de x, \de' x\rangle\rangle ))_{2} & =  
\langle \de x, 
\de f( \langle y,x,x\rangle, \langle \B 0_{y}, \B 0_{x}, \ominus \de' x\rangle )\\
& +
c'(y,x,\B 0_{x})\\ &
+
\de g(\langle y,x,x\rangle, \langle \B 0_{y}, \B 0_{x}, \de' x\rangle )
\rangle
\in \Delta_{W}(x,x')\times \Delta_{Z}(f(y,x,x'),g(y,x,x'))
\end{align*}

Under the assumption that all sets $\Delta_{X}(x,x')$ are singletons, 
we have that $((j\circ r_{Y,X\mid Q})(\langle\langle y,x\rangle, \de x\rangle))_{2}= \langle \de x, c'(y,x, \de x )$ (since it must be $\de x=\B 0_{x}$). Moreover, when $Z=X$, $f(x,x')=x$, $g(x,x')=x'$ and $c'(y,x,\de x)=\B 0_{x}$, we have that 
$(j(\langle \langle y,x,x'\rangle, \langle \de x, \de 'x\rangle\rangle  )_{2}=\langle \de x, \de' x\rangle$ (since $\de g(\langle x,x\rangle, \langle \de x,\de' x\rangle)=\de' x$).

\end{itemize}

\subparagraph*{Extensionality}
The CS model of $\DTT$ satisfies the extensionality axiom \eqref{cext}, together with the equational rules \eqref{jca} and \eqref{jcb} and \eqref{jc}. In fact, 
given simple types $A,B$, the interpretation of $D_{A\times B}$ is generated by the cartesian product in $\DLR$ 
of the interpretations of $D_{A}$ and $D_{B}$, which coincides with the product in $\D P(\model A)$.

The CS model does not satisfy either \eqref{fext1} or \eqref{fext2}; indeed, 
if the CS $X$ and $Y$ interpret two simple types $A$ and $B$, it seems that the exponential of $X$ and $Y$ in $\CS$ cannot  captured by a type of $\DTT$.

%
%
%
%
%
%

\section{Cartesian Differential Categories: Details}

A \emph{cartesian differential category} \cite{Blute2009} (in short, CDC) is a left-additive cartesian category $\BB C$ such that for all 
arrow $f: X\to Y$ there exists an arrow $\de f: X\times X\to Y$ satisfying the axioms below:
\begin{description}
\item[D1.] $\de(f+g)=\de f+\de g$, $\de 0=0$;
\item[D2.] $\de f\circ \langle h+k,v\rangle=\de f\circ \langle h,v\rangle+\de f\circ \langle k,v\rangle$, and $\de f\circ \langle 0,v\rangle=0$;
\item[D3.] $\de(\mathrm{id})=\pi_{1}$, $\de(\pi_{1})=\pi_{1}\circ \pi_{1}$, $\de(\pi_{2})=\pi_{2}\circ \pi_{1}$;
\item[D4.] $\de(\langle f,g\rangle)= \langle \de f, \de g\rangle$;
\item[D5.] $\de(g\circ f)=\de g\circ \langle \de g, g\circ \pi_{2}\rangle$;
\item[D6.] $\de(\de f)\circ \langle \langle g,0\rangle, \langle h,k\rangle\rangle=
\de f\circ \langle g,k\rangle$;
\item[D7.] $\de(\de f))\circ \langle \langle 0,h\rangle, \langle g,k\rangle\rangle = \de(\de f))\circ \langle \langle 0,g\rangle,\langle h,k\rangle\rangle$.
\end{description}
For an intuitive explanation of the axioms see \cite{Blute2009}.

A \emph{differential $\lambda$-category} \cite{Manzo2010} is a CDC which is also a cartesian closed category, and where $\de f$ further satisfies the axiom below:

\begin{description}

\item[D-curry.] $\de(\lambda(f))= \lambda(\de f\circ \langle \pi_{1}\times 0, \pi_{2}\times \mathrm{Id}\rangle)$.
\end{description}

%

We now describe the $\DTT$-structure associated with the identity functor $\mathrm{id}_{\BB C}:\BB C\to \BB C$ on a differential $\lambda$-category $\BB C$.

We let $X^{n}$ be a shorthand for $X\times \dots \times X$ $n$ times. 
For any object $X$, $\D P^{\flat}(X)$ contains all $X^{n}$, with associated object $X^{n+2}$ and projection $\pi_{X}:X^{n+2}\to X\times X$;  $\D P(X)$ is made of all objects of the form $X\times C$, with associated projection $\pi_{X}: X\times C\to X$.
For all $f,g:Y\to X$, the pullback $\langle f,g\rangle^{\sharp}X\times C$ is just $Y\times C$.

An arrow in the slice category $\BB C_{\D P}^{X}$ between  $X\times C$ and $X\times D$ is an arrow $h: X\times C\to D$;
one can check then that:
\begin{itemize}
\item the product of $X\times C$ and $X\times D$ in $\BB C_{\D P}^{X}$ is $X\times (C\times D)$;
moreover, for all pure predicate $X^{n+2}, X^{m+2}\in \D P^{\flat}(X)$, their product
$X^{n+m+2}\in \D P^{\flat}(X)$;%

\item for all objects $X,I$
and pure predicate $X^{n+2}\in \D P^{\flat}(X)$, the dependent product $\Pi_{I}(X^{n+2})$ is the pure predicate $(X^{I})^{n+2}\in \D P^{\flat}(X^{I})$; in fact, by the cartesian closure of $\BB C$ we have that  
$\BB C(\langle \pi_{W},\pi_{W}\rangle^{\sharp}(W^{n+2}),  X^{n+2}   ) \simeq \BB C( W^{n+2}, (X^{I})^{n+2})$, where $\pi_{W}:W\times I\to W$;

\item for all objects $X$ and pure predicate $Y^{n+2}\in \D P^{\flat}(Y)$, their dependent  product is the pure predicate $(Y^{X\times X})^{n+2}\in \D P^{\flat}(Y^{X\times X})$; in fact, by the cartesian closure of $\BB C$ we have 
$\BB C( \langle\pi_{W},\pi_{W}\rangle^{\sharp}W^{n+2},  Y^{n+2}   ) \simeq \BB C( W^{n+2}, (Y^{X\times X})^{n+2})$, where $\pi_{W}:W\times X\times X\to W$.

\end{itemize}

The difference structure is as follows:
\begin{itemize}

\item for all objects $X$, $\Dist X$ is $X^{3}\in \D P(X)$;
\item for all objects $X,Y$ and predicate $Q=V\times C\in \D P(Y\times X\times X)$, where $V=Y\times X\times X$, observe that $\delta_{X}^{\sharp}Q= (Y\times X)\times C$; we let then $r_{Y,X\mid Q}: 
(Y\times X)\times C \to( Y\times X^{2})\times (Q\times X)$ be  $
\mu
\times (\mathrm{id}_{Q}\times 0)$, where $\mu : Y\times X \to Y\times X^{2}$ is $\mathrm{id}_{Y}\times \delta_{X}$;

\item for all objects $X,Y,Z$, predicate $Q=(Y\times X)\times C\in \D P(Y\times X)$, 
pure predicate $P=Z^{n+2}\in \D P^{\flat}(Z)$, arrows $f,g:Y\times X\times X\to Z$, the pullback $\langle f,g\rangle^{\sharp}P$ is $(Y\times X^{2})\times Z^{n}$; for all morphisms
$c: (Y \times X)\times Q \to (Y\times X^{2})\times (Q\times Z^{n})$ such that $c= \langle \mu, \langle \mathrm{id}_{Q}\circ \pi_{2}, c'\rangle\rangle$ for some $c':(Y\times X)\times Q\to Z^{n}$, we can define a diagonal filler $j: (Y\times X^{2})\times (Q\times X)\to (Y\times X^{2})\times (Q\times Z^{n})$ by 
$$
j_{2}=\Big \langle \mathrm{id}_{Q}\circ (\pi_{1}\circ \pi_{2}) , 
c' \circ \nu \circ \pi_{1} +
\big \langle\de f\circ \langle \pi_{1}, \langle\langle 0, \pi_{2}\circ \pi_{2}\rangle, \langle 0,0\rangle \rangle\big\rangle^{n}\rangle
\Big\rangle
$$
where $\nu: Y\times X^{2}\to Y\times X$ is $\mathrm{id}_{Y}\times \pi_{1}$ and $\langle u\rangle^{n}=\langle u,u,\dots, u\rangle$ for $n$ times.

We have that $(j\circ r_{Y,X\mid Q})_{2}=\langle \mathrm{id}_{Q}\circ \pi_{1} ,c'\rangle$ (using the fact that  $\nu\circ \mu=\mathrm{id}_{Y\times X} $). Moreover, when $Z=X$, $n=1$, $f=\pi_{2}\circ \pi_{1}$, $g=\pi_{2}\circ \pi_{2}$ and $c'=0$ we have the $j_{2}= \mathrm{id}_{Q\times X}\circ \pi_{1}$. Hence both the $\beta$- and $\eta$-rules are satisfied.
\end{itemize}

\subparagraph*{Extensionality}
The CDC model of $\DTT$ satisfies the extensionality axioms \eqref{cext} and \eqref{fext1}. This follows from the fact that, if simple types $A,B$ are interpret as objects $X,Y$,
$D_{A\times B}$ is interpreted by the pure predicate $X\times Y\in \D P^{\flat}(X\times Y)$, and $D_{A\to B}$ is interpreted by the pure predicate $Y^{X}\in \D P^{\flat}(Y^{X})$.

Moreover, derivatives satisfy the equational rules \eqref{jca}, \eqref{jcb}, \eqref{jc}, as well as 
\eqref{jl1a} and \eqref{jl1}.

\end{document}